\documentclass[12pt, draftclsnofoot, onecolumn]{IEEEtran}
\usepackage[dvips]{graphicx}
\usepackage{amsmath,amssymb}
\usepackage{cite}
\usepackage{amscd,amsfonts}
\usepackage{graphicx}
\usepackage{subfigure}
\usepackage{color,url}

\newtheorem{lemma}{Lemma}
\begin{document}
\title{Beamforming Design for Full-Duplex Two-Way Amplify-and-Forward MIMO Relay}

\author{Yeonggyu Shim, \emph{Student Member, IEEE}, Wan Choi, \emph{Senior Member, IEEE}, and Hyuncheol Park, \emph{Senior Member, IEEE}\thanks{Y. Shim, W. Choi, and H. Park are with the School of Electrical Engineering, Korea Advanced Institute of Science and Technology (KAIST), 291 Daehak-ro, Yuseong-gu, Daejeon 34141, Korea (e-mail: ygshim@kaist.ac.kr; wchoi@kaist.edu; hcpark@kaist.ac.kr).}}

\maketitle
\vspace{-0.5in}
\begin{abstract}
We consider the full-duplex (FD) two-way amplify-and-forward relay system with imperfect cancelation of loopback self-interference (SI) and investigate joint design of relay and receive beamforming for minimizing the mean square error under a relay transmit power constraint. Due to loopback channel estimation error and limitation of analog-to-digital converter, the loopback SI cannot be completely canceled. Multiple antennas at the relay can help loopback SI suppression but beamforming is required to balance between the residual SI suppression and the desired signal transmission. Moreover, the relay beamforming matrix should be updated every time slot because the residual SI in the previous time slot is amplified by the current beamforming matrix and added to the received signals from the two sources in the current time slot. We derive the optimally balanced relay beamforming and receive beamforming matrices in closed form based on minimum mean square error, taking into account the propagation of the residual loopback SI from the first to the current time slot. We also propose beamforming design using only the channels of the $m$ latest time slots, not from the first time slot. Based on our numerical results, we also identify when FD is beneficial and propose selection between FD and half-duplex according to signal-to-noise ratio and interference-to-noise ratio. 
\end{abstract}


\IEEEpeerreviewmaketitle

\section{Introduction}
\PARstart{T}{wo-way} relaying systems achieve higher spectral efficiency than one-way relaying systems. 
In two-way relay channel (TWRC), two source nodes exchange messages with each other via a relay, using various multiple access broadcast (MABC) protocols consisting of the multiple access (MAC) phase and the broadcast (BC) phase \cite{Laneman, Rankov1, Rankov2, Tarokh1, Tarokh2}. 
Basically, in the MAC phase, two source nodes simultaneously send their signals to the relay node, and then the relay node amplifies and forwards the signal received in the MAC phase to the two source nodes in the BC phase. With a prior knowledge of its own transmitted signal, each source node is able to cancel the propagated self-interference (SI) back to it and decode the desired signal.

Because full-duplex (FD) systems allow concurrent transmission and reception in the same frequency band, they usually achieve higher spectral efficiency than half-duplex (HD) systems only if the loopback SI due to the concurrent transmission and reception is properly suppressed. Using a circulator which connects antenna to transceiver, a single antenna can be shared between the transmitter and the receiver in FD systems. 

The merits of FD systems and two-way relaying systems in terms spectral efficiency motivated two-way relaying with a FD relay which outperforms both two-way relaying with a HD relay and one-way relaying with a FD relay. In \cite{TWRC_FD_no_SI, TWRC_FD_power_allocation, TWRC_FD_relay_mode, TWRC_FD_outage_probability}, the two-way FD amplify-and-forward (AF) relay systems with two sources and one relay with a single antenna each were studied. Power allocation to maximize the achievable rate was explored in \cite{TWRC_FD_no_SI} but 
residual SI at each node was not considered assuming perfect SI suppression. However, given loopback channel estimation errors and limitation of analog-to-digital converter (ADC) in practical environments, the assumption of perfect SI suppression is too ideal to evaluate practical gains of two-way FD relaying. 
Contrary to \cite{TWRC_FD_no_SI}, power allocation addressing residual SI due to imperfect SI cancelation was investigated under a sum power constraint in \cite{TWRC_FD_power_allocation}. To maximize the sum rate in the presence of residual SI, power allocation under individual power constraints and relay mode selection, among one-way HD, two-way HD, one-way FD, and two-way FD, were proposed in \cite{TWRC_FD_relay_mode}. As another performance metric in the presence of residual SI, in \cite{TWRC_FD_outage_probability}, the outage probability of the two-way FD AF relay system was derived. The application of physical layer network coding to the two-way FD AF relay was investigated in terms of bit error rate \cite{TWRC_FD_PNC} when there exists residual SI. A two-way FD AF relay system with multiple relays was addressed in \cite{TWRC_FD_relay_selection} and an optimal relay selection scheme was proposed considering residual SI.

Although residual SI severely hurts the merit of FD, addressing the residual SI only with single antenna at a FD AF relay has limitations. On the other hand, beamforming with multiple antennas enables to suppress the residual SI, which motivated the study of a two-way FD AF relay system with a relay having multiple antennas. However, to the authors' best knowledge, there are only few studies on the two-way FD AF relay system with multiple antennas. In \cite{TWRC_FD_joint_beamforming}, a two-way FD AF relay system  with multiple antennas was studied, where two sources have a single antenna each and the relay has multiple antennas. To maximize the achievable sum rate, beamforming at relay and power allocation at sources were jointly optimized with a zero forcing (ZF) constraint such that the residual loopback SI is perfectly canceled at the relay. However, a closed form solution was not available and thus an iterative algorithm was proposed in \cite{TWRC_FD_joint_beamforming}.  

Contrary to the beamforming matrix which perfectly cancel out the residual SI with a ZF constraint, beamforming at a FD relay is required to balance between the residual SI suppression and the desired signal transmission. In this context,
in a two-way FD AF relay system with multiple antennas with circulator at the relay as well as the two sources, we design the beamforming matrix balancing them in terms of minimum mean square error (MMSE) under a power constraint at the relay.
 Because multiple antennas at the two sources are considered unlike \cite{TWRC_FD_joint_beamforming}, the beamforming matrix also need to be designed to address the inter-stream interference in the multi-input multi-output (MIMO) channels.
Moreover, to optimally balance the residual SI suppression and the desired signal transmission, the beamforming matrix should be updated every time slot because the residual SI in the previous time slot is amplified by the current beamforming matrix and added to the received signals from the two sources in the current time slot.

In this paper, we analytically identify this coupled effect and show that the optimal beamforming matrix at each time slot is determined as a function of not only the current channels but also all the channels in the past time slots from the first transmission time slot. Using the Lagrangian method, we derive the optimized relay beamforming matrix and receive beamforming matrices at sources in closed form. To account for beamforming design with limited memory, 
we also propose a beamforming method using the channels in a limited number of the latest time slots, not from the first time slot, and show that the performance degradation is marginal relative to the beamforming method based on all the channel matrices from the first time slot to the current time slot. 
Finally, from our numerical results, we identify the signal-to-noise ratio (SNR) and interference-to-noise ratio (INR) region where the achievable sum rate of the proposed scheme is greater than that of the HD system. Based on this result, we propose a duplex mode selection scheme in terms of achievable sum rate according to SNR and INR.

The rest of this paper is organized as follows.
Section II describes the system model. In Section III, we derive the optimal beamforming matrix in terms of MMSE and explore beamforming matrix design with limited memory size. In Section IV, we present numerical results. Finally, conclusions are drawn in Section V.

\emph{Notations}:
Matrices and vectors are denoted, respectively, by uppercase and lowercase boldface characters (e.g., \textbf{A} and \textbf{a}).
The transpose, Hermitian, and inverse of \textbf{A} are denoted, respectively, by $\textbf{A}^T$, $\textbf{A}^{H}$, and $\textbf{A}^{-1}$.
An $N \times N$ identity matrix and $N \times N$ matrix consisting of all zero entries are denoted, respectively, by $\textbf{I}_N$ and $\textbf{0}_N$.
The operators $\mathbb{E}[\cdot]$, $\text{tr}(\cdot)$, $\text{det}(\cdot)$, $\text{vec}(\cdot)$, $\text{mat}(\cdot)$, $\otimes$, and $\prod$ indicate the expectation, trace of a matrix, determinant of a matrix, matrix vectorization, inverse operation of $\text{vec}(\cdot)$, Kronecker product, and sequence product operators, respectively.
Notations $||\textbf{a}||$ and $||\textbf{A}||_{\mathcal F} $ denotes 2-norm of \textbf{a} and the Frobenius norm of \textbf{A}, respectively.

\section{System model}

\begin{figure} [t!]
\centering
\includegraphics[width=14.0cm]{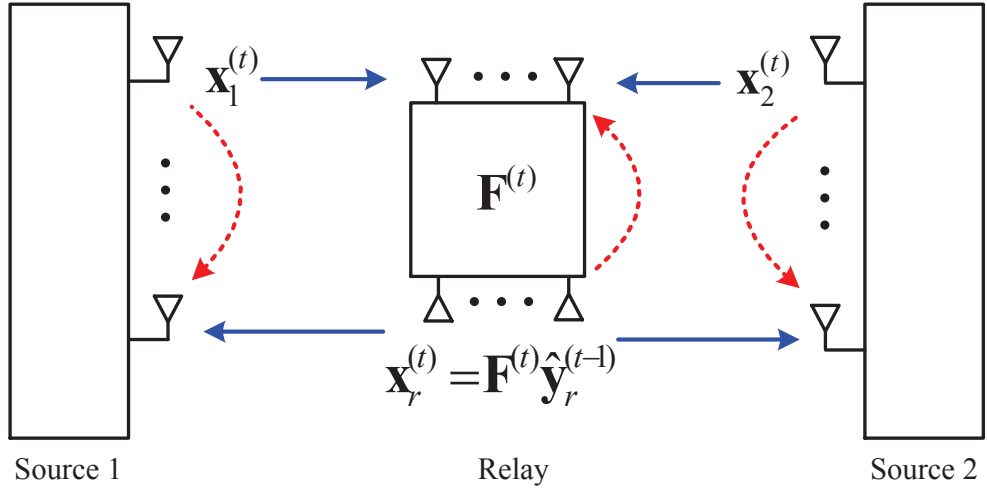}
\caption{Full-duplex two-way amplify-and-forward relay system in time slot $t$.}
\label{fig_system}
\end{figure}

As shown in Fig. \ref{fig_system}, consider a two-way FD AF relaying system consisting of two source nodes $s_1$ and $s_2$ with $N_s$ antennas, and one relay node $r$ with $N_r$ antennas, where all nodes operate in the FD mode. 
In time slot 0, the two sources transmit their signals to the relay simultaneously.
In time slot 1, the relay forwards its received signal after multiplying it by a beamforming matrix to the source nodes and two source nodes transmit their next signals to the relay at the same time.
In this manner, the two source nodes repeatedly exchange their information with each other via relay. In the $t$th time slot, source 1, source 2, and the relay transmit their signals $\textbf{x}_1^{(t)} \in \mathbb{C}^{N_s \times 1}$, $\textbf{x}_2^{(t)} \in \mathbb{C}^{N_s \times 1}$, and $\textbf{x}_r^{(t)} \in \mathbb{C}^{N_r \times 1}$, respectively, of which covariance matrices are given by
$\mathbb{E}[\textbf{x}_i^{(t)}\textbf{x}_i^{(t)^H}]=p_i\textbf{I}_{N_s}$ for $i \in \{ 1,2 \}$ and $\mathbb{E}[\textbf{x}_r^{(t)}\textbf{x}_r^{(t)^H}]=p_r\textbf{I}_{N_r}$.
The noise signals at source $1$, source $2$, and the relay are denoted as $\textbf{n}_i^{(t)} \in \mathbb{C}^{N_s \times 1}$ for $i \in \{ 1,2 \}$ and $\textbf{n}_r^{(t)} \in \mathbb{C}^{N_r \times 1}$, respectively, and are assumed to be zero-mean white Gaussian noise (AWGN) with $\mathbb{E}[\textbf{n}_i^{(t)}\textbf{n}_i^{(t)^H}]=\sigma_{n,i}^2\textbf{I}_{N_s}$ for $i \in \{ 1,2 \}$ and $\mathbb{E}[\textbf{n}_r^{(t)}\textbf{n}_r^{(t)^H}]=\sigma_{n,r}^2\textbf{I}_{N_r}$.

The channel links are modeled as independent and frequency-flat fading channels and assumed to be static in each time slot.
In time slot $t$, the channel matrices between source $1$ and the relay, source $2$ and the relay, the relay and source $1$, and the relay and source $2$ are represented as
$\textbf{H}_{1,r}^{(t)}, \textbf{H}_{2,r}^{(t)} \in \mathbb{C}^{N_r \times N_s}$ and $\textbf{H}_{r,1}^{(t)}, \textbf{H}_{r,2}^{(t)} \in \mathbb{C}^{N_s \times N_r}$, respectively.
Also, the loopback channel matrices at source 1, source 2, and the relay are represented as $\textbf{H}_{i,i}^{(t)} \in \mathbb{C}^{N_s \times N_s}$ for $i \in \{ 1,2 \}$ and $\textbf{H}_{r,r}^{(t)} \in \mathbb{C}^{N_r \times N_r}$, respectively.

The channel state information (CSI) of the loopback channels is assumed to be imperfect due to channel estimation errors, whereas, to focus only on the residual loopback SI related to full-duplex operation, CSI between two nodes is assumed to be perfectly known at each node. 
Since each node knows its own transmitted signals and perfect CSI between two nodes, the backward propagated SI can be canceled perfectly.
However, due to imperfect CSI of the loopback channels, each node cannot perfectly cancel their loopback SI.

In a FD relay system, a loopback channel exists at each node.
The imperfect loopback channel estimation results in loopback SI, but the imperfect backward channel (i.e., channel between two sources via relay) estimation errors are not related to loopback SI. To focus on the effect of loopback SI in FD relay systems, we consider imperfect loopback channel estimation but perfect backward channel estimation is assumed.
Moreover, as the number of involved time slots increase, the impact of the loopback channel estimation errors dominates that of the backward channel estimation errors.

The relationship between the exact channel and the estimated channel is modeled as $\textbf{H}_{i,i}^{(t)}=\hat{\textbf{H}}_{i,i}^{(t)}+\boldsymbol{\Delta}_{i,i}^{(t)}$ for $i \in \{ 1,2,r \}$, where $\hat{\textbf{H}}_{j,j}^{(t)} \in \mathbb{C}^{N_s \times N_s}$ for $j \in \{ 1,2 \}$ and $\hat{\textbf{H}}_{r,r}^{(t)} \in \mathbb{C}^{N_r \times N_r}$ are the estimated channel matrices, and $\boldsymbol{\Delta}_{j,j}^{(t)} \in \mathbb{C}^{N_s \times N_s}$ and $\boldsymbol{\Delta}_{r,r}^{(t)} \in \mathbb{C}^{N_r \times N_r}$ are channel estimation error matrices.
It is assumed that $\boldsymbol{\Delta}_{i,i}^{(t)}$ is Gaussian distributed with zero mean, and
$\mathbb{E}[\text{vec}(\boldsymbol{\Delta}_{j,j}^{(t)})\text{vec}(\boldsymbol{\Delta}_{j,j}^{(t)^H})]=\sigma_{e,j}^2\textbf{I}_{N_s^2}$
and
$\mathbb{E}[\text{vec}(\boldsymbol{\Delta}_{r,r}^{(t)})\text{vec}(\boldsymbol{\Delta}_{r,r}^{(t)^H})]=\sigma_{e,r}^2\textbf{I}_{N_r^2}$.
Also, $\boldsymbol{\Delta}_{i,i}^{(t)}$ for $t \geq 1$ are independent of each other.

The transmission process starts from time slot 0, so the sources transmit their signals but the relay do not in time slot 0. Because the relay starts transmitting from time slot 1, FD operation substantially starts from time slot 1.
In time slot $0$, the two source nodes simultaneously broadcast their signals
$\textbf{x}_1^{(0)} \in \mathbb{C}^{N_s \times 1}$ and $\textbf{x}_2^{(0)} \in \mathbb{C}^{N_s \times 1}$ to the relay and then
the received signal $\textbf{y}_r^{(0)} \in \mathbb{C}^{N_r \times 1}$ at the relay becomes
\begin{align}
\textbf{y}_r^{(0)} = \textbf{H}_{1,r}^{(0)}\textbf{x}_1^{(0)} + \textbf{H}_{2,r}^{(0)}\textbf{x}_2^{(0)} + \textbf{n}_r^{(0)}.
\label{y_r_0}
\end{align} In time slot $1$, the relay forwards its received signal to the source nodes after multiplying it  by a relay beamforming matrix $\textbf{F}^{(1)} \in \mathbb{C}^{N_r \times N_r}$ and then the transmit signal $\textbf{x}_r^{(1)} \in \mathbb{C}^{N_r \times 1}$ at the relay is represented as
\begin{align}
\textbf{x}_r^{(1)} = \textbf{F}^{(1)}\textbf{y}_r^{(0)}
\label{x_r_1}
\end{align} of which transmit power is given by $\mathbb{E}[||\textbf{x}_r^{(1)}||^2]=N_rp_r$.  With imperfect loopback SI cancellation, the received signal at the relay in time slot $1$ is given by
\begin{align}
\hat{\textbf{y}}_r^{(1)} = \textbf{H}_{1,r}^{(1)}\textbf{x}_1^{(1)} + \textbf{H}_{2,r}^{(1)}\textbf{x}_2^{(1)} + \boldsymbol{\Delta}_{r,r}^{(1)}\textbf{x}_r^{(1)} + \textbf{n}_r^{(1)}.
\label{y_r_1}
\end{align}
Substituting \eqref{x_r_1} into \eqref{y_r_1} leads to
\begin{align}
\hat{\textbf{y}}_r^{(1)} =
\textbf{H}_{1,r}^{(1)}\textbf{x}_1^{(1)} + \textbf{H}_{2,r}^{(1)}\textbf{x}_2^{(1)} + \textbf{n}_r^{(1)} + \boldsymbol{\Delta}_{r,r}^{(1)}\textbf{F}^{(1)}(\textbf{H}_{1,r}^{(0)}\textbf{x}_1^{(0)} + \textbf{H}_{2,r}^{(0)}\textbf{x}_2^{(0)} + \textbf{n}_r^{(0)}).
\label{}
\end{align}

At source $l$ in time slot $1$, subtracting the propagated SI back to the source node $\textbf{H}_{r,l}^{(1)}\textbf{F}^{(1)}\textbf{H}_{\bar{l},r}^{(0)}\textbf{x}_l^{(0)}$,  and the loopback SI $\hat{\textbf{H}}_{l,l}^{(1)}\textbf{x}_l^{(1)}$, the received signal at source $l$ is given by
\begin{align}
\hat{\textbf{y}}_l^{(1)} = \textbf{H}_{r,l}^{(1)}\textbf{F}^{(1)}\textbf{H}_{\bar{l},r}^{(0)}\textbf{x}_{\bar{l}}^{(0)} + \textbf{H}_{r,l}^{(1)}\textbf{F}^{(1)}\textbf{n}_r^{(0)} + \boldsymbol{\Delta}_{l,l}^{(1)}\textbf{x}_l^{(1)}  + \textbf{n}_l^{(1)},
\label{}
\end{align}
where $\bar{l}$ means the index of the other source (i.e., $\bar{l}=2$ if $l=1$ and $\bar{l}=1$ if $l=2$).
Source $l$ multiplies $\hat{\textbf{y}}_l^{(1)}$ by a receive beamforming matrix $\textbf{R}_l^{(1)} \in \mathbb{C}^{N_s \times N_s}$.

Accordingly, considering imperfect loopback SI cancellation, the received signal at the relay in time slot $t-1$ ($t \geq 2$) is given by
\begin{align}
\hat{\textbf{y}}_{r}^{(t-1)} =
\textbf{H}_{1,r}^{(t-1)}\textbf{x}_1^{(t-1)} + \textbf{H}_{2,r}^{(t-1)}\textbf{x}_2^{(t-1)} + \textbf{n}_r^{(t-1)} 
+ \sum_{i=0}^{t-2}  \bigg\{  \prod_{j=1}^{t-1-i} ( \boldsymbol{\Delta}_{r,r}^{(t-j)}\textbf{F}^{(t-j)} )
(\textbf{H}_{1,r}^{(i)}\textbf{x}_1^{(i)} + \textbf{H}_{2,r}^{(i)}\textbf{x}_2^{(i)} + \textbf{n}_r^{(i)} )  \bigg\}
\label{resursive}
\end{align} which shows that the current received signal is affected by the beamforming matrices in the past time slots. 
In time slot $t$, the forwarded signal from the relay to the source nodes is given by
\begin{align}
\textbf{x}_r^{(t)} = \textbf{F}^{(t)}\hat{\textbf{y}}_r^{(t-1)}
\label{}
\end{align} where  $\textbf{F}^{(t)} \in \mathbb{C}^{N_r \times N_r}$ is the relay beamforming matrix at time slot $t$ and $\textbf{x}_r^{(t)} \in \mathbb{C}^{N_r \times 1}$. Subtracting the propagated SI back to the source node $\textbf{H}_{r,l}^{(t)}\textbf{F}^{(t)}\textbf{H}_{\bar{l},r}^{(t-1)}\textbf{x}_l^{(t-1)}$ and the loopback SI $\hat{\textbf{H}}_{l,l}^{(t)}\textbf{x}_l^{(t)}$, the received signal at source $l$ in time slot $t$ ($t \geq 2$) is obtained as
\begin{equation}
\begin{split}
\hat{\textbf{y}}_l^{(t)} &= \textbf{H}_{r,l}^{(t)}\textbf{F}^{(t)}\textbf{H}_{\bar{l},r}^{(t-1)}\textbf{x}_{\bar{l}}^{(t-1)}
+ \textbf{H}_{r,l}^{(t)}\textbf{F}^{(t)}\sum_{i=0}^{t-2}  \bigg\{  \prod_{j=1}^{t-1-i} ( \boldsymbol{\Delta}_{r,r}^{(t-j)}\textbf{F}^{(t-j)} )
(\textbf{H}_{l,r}^{(i)}\textbf{x}_l^{(i)} + \textbf{H}_{\bar{l},r}^{(i)}\textbf{x}_{\bar{l}}^{(i)} + \textbf{n}_r^{(i)})  \bigg\}
\\& + \textbf{H}_{r,l}^{(t)}\textbf{F}^{(t)}\textbf{n}_r^{(t-1)} +
\boldsymbol{\Delta}_{l,l}^{(t)}\textbf{x}_l^{(t)} + \textbf{n}_l^{(t)}.
\label{}
\end{split}
\end{equation}
After subtracting SI, source $l$ multiplies $\hat{\textbf{y}}_l^{(t)}$ by a receive beamforming matrix $\textbf{R}_l^{(t)} \in \mathbb{C}^{N_s \times N_s}$.

\section{MMSE based beamforming design}

The problem for joint design of relay beamforming matrices and source receive beamforming matrices is non-convex and basically intractable. Therefore, we propose an iterative algorithm which decouples the primal problem into two subproblems and solve them alternately; one is for relay beamforming design and the other is for receive beamforming design at sources.

As identified in \eqref{resursive}, the current received signal is a function of  the beamforming matrices, the channel estimation error matrices, and transmitted signals at source in the past time slots, which implies that to optimally balance the residual SI suppression and the desired signal transmission, the beamforming matrix should be updated every time slot. Without loss of generality,  we design the relay beamforming matrix based on the $m$ latest time slots; if $m = t-1$, the beamforming matrix in time slot $t$ is based on all the time slots from the first time slot. Otherwise if $m< t-1$, the beamforming design corresponds to the scenario when the memory size at the relay is limited.
Unlike one-way relay systems, transmitted signals at each source during the past time slots become interference signals at each source in the current time slot in two-way relay systems. In addition, the relay transmits the received noise signals during the past time slots to the sources. Therefore, the minimum mean square error (MMSE) criterion is an effective metric enabling a compromise between interference suppression and noise reduction. Moreover, the MMSE criterion allows tractable analysis and optimization of the two-way FD AF relay system. For similar reasons, mean square error (MSE) has been popularly used for beamforming design in two-way relay systems in the literature such as \cite{lemma1_equation, MSE1, MSE2, scaling_factor1, MSE3, MSE4}.

To balance the residual SI suppression and the desired signal transmission, the relay beamforming matrix $\textbf{F}^{(t)}$ and receive beamforming matrices $\textbf{R}_i^{(t)}$ for $i \in \{ 1,2 \}$ in time slot $t ~(\geq 1)$ are determined based on the MSE criterion under a transmit power constraint at the relay. The corresponding optimization problem is formulated as
\begin{align}
\mathbb{P}^{(t)}:
\mathop{\mbox{min}}_{\textbf{F}^{(t)}, \alpha^{(t)}, \textbf{R}_1^{(t)}, \textbf{R}_2^{(t)}} J(\textbf{F}^{(t)}, \alpha^{(t)}, \textbf{R}_1^{(t)}, \textbf{R}_2^{(t)})
\ \ \text{s.t.}\ \ \mathbb{E}\big[||\textbf{x}_r^{(t)}||^{2}\big] = N_rp_r
\label{}
\end{align}
where the sum of MSE $J(\textbf{F}^{(t)},\alpha^{(t)},\textbf{R}_1^{(t)},\textbf{R}_2^{(t)})$ is defined as
\begin{align}
J(\textbf{F}^{(t)}, \alpha^{(t)}, \textbf{R}_1^{(t)}, \textbf{R}_2^{(t)}) \triangleq
\mathbb{E}\big[||\textbf{x}_2^{(t-1)}- \alpha^{(t)^{-1}}\textbf{R}_1^{(t)^H}\hat{\textbf{y}}_1^{(t)} ||^2 \big]
+ \mathbb{E}\big[||\textbf{x}_1^{(t-1)}- \alpha^{(t)^{-1}}\textbf{R}_2^{(t)^H}\hat{\textbf{y}}_2^{(t)} ||^2 \big]
\label{}
\end{align} and $\alpha^{(t)^{-1}}$ is a positive scaling factor. Note that $\hat{\textbf{y}}_l^{(t)}$ is multiplied by 
the positive scaling factor $\alpha^{(t)^{-1}}$ since the relay beamforming matrix 
$\textbf{F}^{(t)}$ is split into two components, as in \cite{scaling_factor1, scaling_factor2}, such that $\textbf{F}^{(t)}=\alpha^{(t)}\bar{\textbf{F}}^{(t)}$ to facilitate the derivation of a closed form solution, where $||\bar{\textbf{F}}^{(t)}||_{\mathcal F} =1$ and thus $\alpha^{(t)}$ and $\bar{\textbf{F}}^{(t)}$ represent power amplification and steering directions, respectively. 

The Lagrangian of the optimization problem is defined as 
\begin{equation}
\begin{split}
& \mathcal{L}(\textbf{F}^{(t)},\alpha^{(t)},\textbf{R}_1^{(t)},\textbf{R}_2^{(t)},\lambda^{(t)})
\\
&= J(\textbf{F}^{(t)},\alpha^{(t)},\textbf{R}_1^{(t)},\textbf{R}_2^{(t)}) + \lambda^{(t)}(\mathbb{E}\big[||\textbf{x}_r^{(t)}||^{2}\big]-N_rp_r)\\
 &=
N_s(p_1+p_2) - \alpha^{(t)^{-1}} \{ \text{tr}(\textbf{W}_{f,0}^{(t)^H}\textbf{F}^{(t)}) + \text{tr}(\textbf{W}_{f,0}^{(t)}\textbf{F}^{(t)^H}) \}
\\
&\quad+ \alpha^{(t)^{-2}} \{ \mathbb{E}\big[|| \textbf{R}_1^{(t)^H}\hat{\textbf{y}}_1^{(t)} ||^{2}\big] + \mathbb{E}\big[|| \textbf{R}_2^{(t)^H}\hat{\textbf{y}}_2^{(t)} ||^{2}\big] \}
+ \lambda^{(t)}(\mathbb{E}\big[||\textbf{x}_r^{(t)}||^{2}\big]-N_rp_r),
\label{L22}
\end{split}
\end{equation}
where $\lambda^{(t)}$ is a Lagrangian multiplier in time slot $t$, $\textbf{W}_{f,0}^{(t)} = p_1\textbf{H}_{r,2}^{(t)^H}\textbf{R}_2^{(t)}\textbf{H}_{1,r}^{(t-1)^H} + p_2\textbf{H}_{r,1}^{(t)^H}\textbf{R}_1^{(t)}\textbf{H}_{2,r}^{(t-1)^H}$, 

\begin{equation}
\begin{split}
\mathbb{E}\big[|| \textbf{R}_1^{(t)^H}\hat{\textbf{y}}_1^{(t)} ||^{2}\big] &=
\mathbb{E} \big[ \text{tr} \big\{ \textbf{R}_1^{(t)^H}\textbf{H}_{r,1}^{(t)}\textbf{F}^{(t)}
(\textbf{G}_0^{(t)} + p_2\textbf{H}_{2,r}^{(t-1)}\textbf{H}_{2,r}^{(t-1)^H} + \sigma_{n,r}^2\textbf{I}_{N_r})
\textbf{F}^{(t)^H}\textbf{H}_{r,1}^{(t)^H}\textbf{R}_1^{(t)} \big\} \big]
\\
&+ p_1 \mathbb{E} \big[ \text{tr} (\textbf{R}_1^{(t)^H}\boldsymbol{\Delta}_{1,1}^{(t)}\boldsymbol{\Delta}_{1,1}^{(t)^H}\textbf{R}_1^{(t)}) \big]
+ \sigma_{n,1}^2\text{tr}(\textbf{R}_1^{(t)^H}\textbf{R}_1^{(t)}),
\\
\mathbb{E}\big[|| \textbf{R}_2^{(t)^H}\hat{\textbf{y}}_2^{(t)} ||^{2}\big] &=
\mathbb{E}\big[ \text{tr} \big\{ \textbf{R}_2^{(t)^H}\textbf{H}_{r,2}^{(t)}\textbf{F}^{(t)}
(\textbf{G}_0^{(t)} + p_1\textbf{H}_{1,r}^{(t-1)}\textbf{H}_{1,r}^{(t-1)^H} + \sigma_{n,r}^2\textbf{I}_{N_r})
\textbf{F}^{(t)^H}\textbf{H}_{r,2}^{(t)^H}\textbf{R}_2^{(t)} \big\} \big]
\\
&+ p_2 \mathbb{E} \big[ \text{tr} (\textbf{R}_2^{(t)^H}\boldsymbol{\Delta}_{2,2}^{(t)}\boldsymbol{\Delta}_{2,2}^{(t)^H}\textbf{R}_2^{(t)}) \big]
+ \sigma_{n,2}^2\text{tr}(\textbf{R}_2^{(t)^H}\textbf{R}_2^{(t)}),
\\
\mathbb{E}\big[|| \textbf{x}_r^{(t)} ||^{2}\big] &=
\mathbb{E}\big[ \text{tr} \big\{ \textbf{F}^{(t)}
(\textbf{G}_0^{(t)} + p_1\textbf{H}_{1,r}^{(t-1)}\textbf{H}_{1,r}^{(t-1)^H} + p_2\textbf{H}_{2,r}^{(t-1)}\textbf{H}_{2,r}^{(t-1)^H} + \sigma_{n,r}^2\textbf{I}_{N_r})
\textbf{F}^{(t)^H}\big\} \big],
\label{exp_y1_y2_xr}
\end{split}
\end{equation} where
$\textbf{G}_0^{(1)}=\textbf{0}_{N_r}$,
\begin{equation}
\begin{split}
\textbf{G}_0^{(t)}
= &\sum_{i=0}^{t-2}  \bigg\{  \prod_{j=1}^{t-1-i}( \boldsymbol{\Delta}_{r,r}^{(t-j)}\textbf{F}^{(t-j)})
(p_1\textbf{H}_{1,r}^{(i)}\textbf{H}_{1,r}^{(i)^H} + p_2\textbf{H}_{2,r}^{(i)}\textbf{H}_{2,r}^{(i)^H} + \sigma_{n,r}^2\textbf{I}_{N_r})
\prod_{j=1}^{t-1-i}(\textbf{F}^{(i+j)^H}\boldsymbol{\Delta}_{r,r}^{(i+j)^H})
\bigg\}
\\
& \text{for} \ 2 \leq t \leq m+1, ~\textrm{and}
\\
\textbf{G}_0^{(t)}
= &\sum_{i=0}^{t-2-m} \bigg\{
\prod_{j=1}^{m}(\boldsymbol{\Delta}_{r,r}^{(t-j)}\textbf{F}^{(t-j)})
\prod_{j=m+1}^{t-1-i}(\boldsymbol{\Delta}_{r,r}^{(t-j)}\textbf{F}^{(t-m)})
\\
& (p_1\textbf{H}_{1,r}^{(t-1-m)}\textbf{H}_{1,r}^{(t-1-m)^H} + p_2\textbf{H}_{2,r}^{(t-1-m)}\textbf{H}_{2,r}^{(t-1-m)^H} + \sigma_{n,r}^2\textbf{I}_{N_r})
\\
& \prod_{j=1}^{t-1-m-i}(\textbf{F}^{(t-m)^H}\boldsymbol{\Delta}_{r,r}^{(i+j)^H})
\prod_{j=t-m}^{t-1}(\textbf{F}^{(j)^H}\boldsymbol{\Delta}_{r,r}^{(j)^H}) \bigg\}
\\
+ &\sum_{i=t-1-m}^{t-2}  \bigg\{  \prod_{j=1}^{t-1-i}( \boldsymbol{\Delta}_{r,r}^{(t-j)}\textbf{F}^{(t-j)})
(p_1\textbf{H}_{1,r}^{(i)}\textbf{H}_{1,r}^{(i)^H} + p_2\textbf{H}_{2,r}^{(i)}\textbf{H}_{2,r}^{(i)^H} + \sigma_{n,r}^2\textbf{I}_{N_r})
\prod_{j=1}^{t-1-i}(\textbf{F}^{(i+j)^H}\boldsymbol{\Delta}_{r,r}^{(i+j)^H})  \bigg\}
\\
& \text{for} \ t \geq m+2.
\end{split}
\end{equation}
Note that $m$ is the number of past time slots considered in the beamforming matrix design and the 
expectations in \eqref{exp_y1_y2_xr} are with respect to $\boldsymbol{\Delta}_{i,i}^{(t)}$ for $i \in \{ 1,2,r  \}$.

\begin{lemma}
Let $\boldsymbol{\Delta}_j$ be an $N \times N$ random matrix with $\mathbb{E}[\text{vec}(\boldsymbol{\Delta}_j)\text{vec}(\boldsymbol{\Delta}_j)^H]=\sigma^2\textbf{I}_{N^2}$
and $\boldsymbol{\Delta}_j$ for $j \geq 1$ are independent of each other.
\begin{align}
\mathbb{E}\bigg[ \text{tr} \bigg\{
\textbf{V}_v
\prod_{j=1}^{v-1} ( \boldsymbol{\Delta}_{v-j}\textbf{V}_{v-j} )
\prod_{j=1}^{v-1} ( \textbf{V}_j^H\boldsymbol{\Delta}_j^H )
\textbf{V}_v^H
\bigg\} \bigg]
=  (\sigma^2)^{v-1} \prod_{j=1}^{v}\text{tr}(\textbf{V}_j\textbf{V}_j^H)
\ \text{for} \ v \geq 2.
\label{lemma_1}
\end{align}
\end{lemma}
\begin{proof} We prove this by induction.
From \emph{Lemma I} in \cite{lemma_error}, \eqref{lemma_1} holds for $v=2$. That is,
$\mathbb{E} [ \text{tr} \{
\textbf{V}_2\boldsymbol{\Delta}_1\textbf{V}_1\textbf{V}_1^H\boldsymbol{\Delta}_1^H\textbf{V}_2^H
\} ]
= \sigma^2\text{tr}(\textbf{V}_1\textbf{V}_1^H)\text{tr}(\textbf{V}_2\textbf{V}_2^H).$
Assume \eqref{lemma_1} holds when $v=k-1$ and then we prove that it also holds when $v=k$. 

When $v=k$, the left hand side of \eqref{lemma_1} is rewritten as 
\begin{align}
&\text{LHS~of \eqref{lemma_1} } \nonumber\\ & \overset{(a)}{=} \mathbb{E}\bigg[ \text{tr} \bigg\{
\prod_{j=1}^{k-1} (\boldsymbol{\Delta}_{k-j}\textbf{V}_{k-j})
\prod_{j=1}^{k-1} (\textbf{V}_j^H\boldsymbol{\Delta}_j^H)
\textbf{V}_k^H\textbf{V}_k
\bigg\} \bigg]\label{lemma_proof_1}\\
&\overset{(b)}{=}\mathbb{E}\bigg[
\text{vec}(\boldsymbol{\Delta}_{k-1})^H
\bigg\{
\bigg(
\textbf{V}_{k-1}
\prod_{j=1}^{k-2} (\boldsymbol{\Delta}_{k-1-j}\textbf{V}_{k-1-j})
\prod_{j=1}^{k-2} (\textbf{V}_j^H\boldsymbol{\Delta}_j^H)
\textbf{V}_{k-1}^H
\bigg)^T
\otimes
(\textbf{V}_k^H\textbf{V}_k)
\bigg\}
\text{vec}(\boldsymbol{\Delta}_{k-1})
\bigg]\label{lemma_proof_2}\\
&\overset{(c)}{=} \text{tr}\left( \mathbb{E}\bigg[
\text{vec}(\boldsymbol{\Delta}_{k-1})^H
\bigg\{
\bigg(
\textbf{V}_{k-1}
\prod_{j=1}^{k-2} (\boldsymbol{\Delta}_{k-1-j}\textbf{V}_{k-1-j})
\prod_{j=1}^{k-2} (\textbf{V}_j^H\boldsymbol{\Delta}_j^H)
\textbf{V}_{k-1}^H
\bigg)^T
\otimes
(\textbf{V}_k^H\textbf{V}_k)
\bigg\}
\text{vec}(\boldsymbol{\Delta}_{k-1})
\bigg]\right)\\
&\overset{(d)}{=}\text{tr} \bigg\{
\mathbb{E}\bigg[
\bigg(
\textbf{V}_{k-1}
\prod_{j=1}^{k-2} (\boldsymbol{\Delta}_{k-1-j}\textbf{V}_{k-1-j})
\prod_{j=1}^{k-2} (\textbf{V}_j^H\boldsymbol{\Delta}_j^H)
\textbf{V}_{k-1}^H
\bigg)^T
\otimes
(\textbf{V}_k^H\textbf{V}_k)
\bigg]
\mathbb{E}\bigg[ \text{vec}(\boldsymbol{\Delta}_{k-1}) \text{vec}(\boldsymbol{\Delta}_{k-1})^H \bigg]
\bigg\}
\label{lemma_proof_3}\\
&\overset{(e)}{=}\sigma^2\text{tr}(\textbf{V}_k\textbf{V}_k^H)
\mathbb{E}\bigg[ \text{tr} \bigg\{
\textbf{V}_{k-1}
\prod_{j=1}^{k-2} (\boldsymbol{\Delta}_{k-1-j}\textbf{V}_{k-1-j})
\prod_{j=1}^{k-2} (\textbf{V}_j^H\boldsymbol{\Delta}_j^H)
\textbf{V}_{k-1}^H
\bigg\} \bigg]
\label{lemma_proof_4}
\end{align}
where (a) is due to $\text{tr}(\textbf{A}\textbf{B})=\text{tr}(\textbf{B}\textbf{A})$; (b) is due to  $\text{tr}(\textbf{A}\textbf{B}\textbf{C}\textbf{D})=\text{vec}(\textbf{C}^T)^T(\textbf{B}^T\otimes\textbf{D})\text{vec}(\textbf{A})$; (c) is because $c =\text{tr}(c)$ when $c$ is a constant; (d) follows from $\text{tr}(\mathbb{E}[\textbf{A}])=\mathbb{E}[\text{tr}(\textbf{A})]$,
$\text{tr}(\textbf{A}\textbf{B})=\text{tr}(\textbf{B}\textbf{A})$, and
independence of $\boldsymbol{\Delta}_j$; (e) holds because  $\text{tr}(\mathbb{E}[\textbf{A}])=\mathbb{E}[\text{tr}(\textbf{A})]$, $\mathbb{E}[\text{vec}(\boldsymbol{\Delta}_{k-1})\text{vec}(\boldsymbol{\Delta}_{k-1})^H]=\sigma^2\textbf{I}_{N^2}$,
$\text{tr}(\textbf{A}\otimes\textbf{B})=\text{tr}(\textbf{A})\text{tr}(\textbf{B})$, and
$\text{tr}(\textbf{A})^T=\text{tr}(\textbf{A})$. 

Finally, using the induction hypothesis that \eqref{lemma_1} holds for $k-1$, \eqref{lemma_proof_4} becomes
\begin{align}
(\sigma^2)^{k-1} \prod_{j=1}^{k}\text{tr}(\textbf{V}_j\textbf{V}_j^H)
\end{align}
\end{proof}

\emph{Remark:} \emph{Lemma 1} provides a mathematical tool to reach the main result. Specifically, to calculate MSE, the expectations on the received signals have to be calculated and \emph{Lemma 1} enables to calculate the expectations since the received signals are constituted by the product of unknown random matrices (i.e., channel estimation error matrices) and known deterministic matrices (i.e., channel matrices and relay beamforming matrices) in rotation, which is the form in \emph{Lemma 1}. Consequently, \emph{Lemma 1} can be applied to the system design in which the received signal is presented by the product of random matrices and deterministic matrices in rotation.

Using \emph{Lemma 1},
$\mathbb{E}\big[|| \textbf{R}_1^{(t)}\hat{\textbf{y}}_1^{(t)} ||^{2}\big]$, $\mathbb{E}\big[|| \textbf{R}_2^{(t)}\hat{\textbf{y}}_2^{(t)} ||^{2}\big]$, and $\mathbb{E}\big[|| \textbf{x}_r^{(t)} ||^{2}\big]$ are obtained, respectively, as
\begin{equation}
\begin{split}
\mathbb{E}\big[|| \textbf{R}_1^{(t)}\hat{\textbf{y}}_1^{(t)} ||^{2}\big] &=
\text{tr} \big\{ \textbf{H}_{r,1}^{(t)^H}\textbf{R}_1^{(t)}\textbf{R}_1^{(t)^H}\textbf{H}_{r,1}^{(t)}\textbf{F}^{(t)}
(\textbf{G}_c^{(t)} + p_2\textbf{H}_{2,r}^{(t-1)}\textbf{H}_{2,r}^{(t-1)^H}+\sigma_{n,r}^2\textbf{I}_{N_r})
\textbf{F}^{(t)^H} \big\}
\\
&+ (N_sp_1\sigma_{e,1}^2 + \sigma_{n,1}^2)\text{tr}(\textbf{R}_1^{(t)^H}\textbf{R}_1^{(t)}),
\\
\mathbb{E}\big[|| \textbf{R}_2^{(t)}\hat{\textbf{y}}_2^{(t)} ||^{2}\big] &=
\text{tr} \big\{ \textbf{H}_{r,2}^{(t)^H}\textbf{R}_2^{(t)}\textbf{R}_2^{(t)^H}\textbf{H}_{r,2}^{(t)}\textbf{F}^{(t)}
(\textbf{G}_c^{(t)} + p_1\textbf{H}_{1,r}^{(t-1)}\textbf{H}_{1,r}^{(t-1)^H}+\sigma_{n,r}^2\textbf{I}_{N_r})
\textbf{F}^{(t)^H} \big\}
\\
&+ (N_sp_2\sigma_{e,2}^2 + \sigma_{n,2}^2)\text{tr}(\textbf{R}_2^{(t)^H}\textbf{R}_2^{(t)}),
\\
\mathbb{E}\big[|| \textbf{x}_r^{(t)} ||^{2}\big] &=
\text{tr} \big\{ \textbf{F}^{(t)}
(\textbf{G}_c^{(t)} + p_1\textbf{H}_{1,r}^{(t-1)}\textbf{H}_{1,r}^{(t-1)^H} + p_2\textbf{H}_{2,r}^{(t-1)}\textbf{H}_{2,r}^{(t-1)^H}+\sigma_{n,r}^2\textbf{I}_{N_r})
\textbf{F}^{(t)^H}\big\},
\label{exp_y1_y2_xr_W}
\end{split}
\end{equation}
where
\begin{align}
\begin{split}
&\textbf{G}_c^{(t)}
\\
= &\mathbf{\Gamma}_1 \sigma_{e,r}^2\textbf{I}_{N_r}\text{tr} \{ \textbf{F}^{(t-1)}(p_1\textbf{H}_{1,r}^{(t-2)}\textbf{H}_{1,r}^{(t-2)^H}+p_2\textbf{H}_{2,r}^{(t-2)}\textbf{H}_{2,r}^{(t-2)^H}+\sigma_{n,r}^2\textbf{I}_{N_r})\textbf{F}^{(t-1)^H} \}
\\
+ &\mathbf{\Gamma}_2 \sum_{i=2}^{\text{min}(m,t-1)} 
\bigg[ (\sigma_{e,r}^2)^i\textbf{I}_{N_r}
\prod_{j=t+1-i}^{t-1} \text{tr}(\textbf{F}^{(j)}\textbf{F}^{(j)^H})
\\
& \text{tr} \{ \textbf{F}^{(t-i)}(p_1\textbf{H}_{1,r}^{(t-i)}\textbf{H}_{1,r}^{(t-i)^H}+p_2\textbf{H}_{2,r}^{(t-i)}\textbf{H}_{2,r}^{(t-i)^H}+\sigma_{n,r}^2\textbf{I}_{N_r})\textbf{F}^{(t-i)^H} \} \bigg]
\\
+ &\mathbf{\Gamma}_3 \sum_{i=2}^{t-1} 
\bigg[ (\sigma_{e,r}^2)^i\textbf{I}_{N_r}
\prod_{j=t+1-i}^{t-1} \text{tr}(\textbf{F}^{(j)}\textbf{F}^{(j)^H})
\\
& \text{tr} \{ \textbf{F}^{(t-i)}(p_1\textbf{H}_{1,r}^{(t-1-i)}\textbf{H}_{1,r}^{(t-1-i)^H}+p_2\textbf{H}_{2,r}^{(t-1-i)}\textbf{H}_{2,r}^{(t-1-i)^H}+\sigma_{n,r}^2\textbf{I}_{N_r})\textbf{F}^{(t-i)^H} \} \bigg]
\label{}
\end{split}
\end{align} and $\mathbf{\Gamma}_1$, $\mathbf{\Gamma}_2$, and $\mathbf{\Gamma}_3$ are defined as
\begin{equation}
\begin{split}
\mathbf{\Gamma}_1=\mathbf{\Gamma}_2=\mathbf{\Gamma}_3=\textbf{0}_{N_r}             \ &\text{for} \ t=1 \\
\mathbf{\Gamma}_1=\textbf{I}_{N_r} \ \text{and} \ \mathbf{\Gamma}_2=\mathbf{\Gamma}_3=\textbf{0}_{N_r} \ &\text{for} \ t=2 \\
\mathbf{\Gamma}_1=\mathbf{\Gamma}_3=\textbf{I}_{N_r} \ \text{and} \ \mathbf{\Gamma}_2=\textbf{0}_{N_r} \ &\text{for} \
t \geq 3 \ \text{and} \ m=1 \\
\mathbf{\Gamma}_1=\mathbf{\Gamma}_2=\textbf{I}_{N_r} \ \text{and} \ \mathbf{\Gamma}_3=\textbf{0}_{N_r} \ &\text{for} \ 3  \leq t \leq m+1 \\
\mathbf{\Gamma}_1=\mathbf{\Gamma}_2=\mathbf{\Gamma}_3=\textbf{I}_{N_r} \ &\text{for} \ t \geq m+2 
\label{}
\end{split}
\end{equation}

\subsection{Relay beamforming design}

The problem of relay beamforming matrix design in each time slot is formulated as
\begin{align}
\mathbb{P}_f^{(t)}:
\mathop{\mbox{min}}_{\textbf{F}^{(t)},\alpha^{(t)}} J(\textbf{F}^{(t)},\alpha^{(t)})
\ \ \text{s.t.}\ \ \mathbb{E}\big[||\textbf{x}_r^{(t)}||^{2}\big] = N_rp_r
\label{}
\end{align}
Using \eqref{exp_y1_y2_xr_W} and substituting $\textbf{F}^{(t)}=\alpha^{(t)}\bar{\textbf{F}}^{(t)}$, the Lagrangian in \eqref{L22} is reduced to
\begin{equation}
\begin{split}
\mathcal{L}(\bar{\textbf{F}}^{(t)},\alpha^{(t)},\lambda^{(t)}) &=
N_s(p_1+p_2) - 
\{ \text{tr}(\textbf{W}_{f,0}^{(t)^H}\bar{\textbf{F}}^{(t)}) + \text{tr}(\textbf{W}_{f,0}^{(t)}\bar{\textbf{F}}^{(t)^H}) \}
\\
&+  
\text{tr}(\textbf{W}_{f,1}^{(t)}\bar{\textbf{F}}^{(t)}\textbf{G}_1^{(t)}\bar{\textbf{F}}^{(t)^H})+
\text{tr}(\textbf{W}_{f,2}^{(t)}\bar{\textbf{F}}^{(t)}\textbf{G}_2^{(t)}\bar{\textbf{F}}^{(t)^H})
\\
&+ \alpha^{(t)^{-2}}w_f^{(t)}
+ \lambda^{(t)} \{ \alpha^{(t)^2} \text{tr}(\bar{\textbf{F}}^{(t)}\textbf{G}_r^{(t)}\bar{\textbf{F}}^{(t)^H})-N_rp_r \},
\label{Lag}
\end{split}
\end{equation}
where
\begin{equation}
\begin{split}
& \textbf{G}_1^{(t)} = \textbf{G}_c^{(t)} + p_2\textbf{H}_{2,r}^{(t-1)}\textbf{H}_{2,r}^{(t-1)^H} + \sigma_{n,r}^2\textbf{I}_{N_r}, \\
& \textbf{G}_2^{(t)} = \textbf{G}_c^{(t)} + p_1\textbf{H}_{1,r}^{(t-1)}\textbf{H}_{1,r}^{(t-1)^H} + \sigma_{n,r}^2\textbf{I}_{N_r}, \\
& \textbf{G}_r^{(t)} = \textbf{G}_c^{(t)} + p_1\textbf{H}_{1,r}^{(t-1)}\textbf{H}_{1,r}^{(t-1)^H} + p_2\textbf{H}_{2,r}^{(t-1)}\textbf{H}_{2,r}^{(t-1)^H} + \sigma_{n,r}^2\textbf{I}_{N_r},  \\
& \textbf{W}_{f,1}^{(t)} = \textbf{H}_{r,1}^{(t)^H}\textbf{R}_1^{(t)}\textbf{R}_1^{(t)^H}\textbf{H}_{r,1}^{(t)}, \\
& \textbf{W}_{f,2}^{(t)} = \textbf{H}_{r,2}^{(t)^H}\textbf{R}_2^{(t)}\textbf{R}_2^{(t)^H}\textbf{H}_{r,2}^{(t)}, \\
& w_f^{(t)} = (N_sp_1\sigma_{e,1}^2+\sigma_{n,1}^2)\text{tr}(\textbf{R}_1^{(t)}\textbf{R}_1^{(t)^H})
          + (N_sp_2\sigma_{e,2}^2+\sigma_{n,2}^2)\text{tr}(\textbf{R}_2^{(t)}\textbf{R}_2^{(t)^H}). \\
\label{G}
\end{split}
\end{equation}

For the Lagrangian in \eqref{Lag}, we have three conditions of optimality as
$\frac{\partial \mathcal{L}(\bar{\textbf{F}}^{(t)}, \alpha^{(t)}, \lambda^{(t)})} {\partial \bar{\textbf{F}}^{(t)}}=\textbf{0}_{N_r}$,
$\frac{\partial \mathcal{L}(\bar{\textbf{F}}^{(t)}, \alpha^{(t)}, \lambda^{(t)})} {\partial \alpha^{(t)}}=0$, and
$\frac{\partial \mathcal{L}(\bar{\textbf{F}}^{(t)}, \alpha^{(t)}, \lambda^{(t)})} {\partial \lambda^{(t)}}=0$. Using the linear and nonlinear properties of the complex matrix derivative \cite{matrix_derivative} and the cyclic permutation and linearity of the trace function, these conditions are reduced to
\begin{align}
& \textbf{W}_{f,1}^{(t)}\bar{\textbf{F}}^{(t)}\textbf{G}_1^{(t)}
+ \textbf{W}_{f,2}^{(t)}\bar{\textbf{F}}^{(t)}\textbf{G}_2^{(t)}
+ \lambda^{(t)}\alpha^{(t)^2}\bar{\textbf{F}}^{(t)}\textbf{G}_r^{(t)}
= \textbf{W}_{f,0}^{(t)}
\label{L1}
\\
& \text{tr}(\bar{\textbf{F}}^{(t)}\textbf{G}_r^{(t)}\bar{\textbf{F}}^{(t)^H})
= \lambda^{(t)^{-1}}\alpha^{(t)^{-4}}w_f^{(t)}
\label{L2}
\\
& \text{tr}(\bar{\textbf{F}}^{(t)}\textbf{G}_r^{(t)}\bar{\textbf{F}}^{(t)^H})
= \alpha^{(t)^{-2}}N_rp_r
\label{L3}
\end{align}
From \eqref{L2} and \eqref{L3}, we have 
\begin{align}
\lambda^{(t)}\alpha^{(t)^2}=(N_rp_r)^{-1}w_f^{(t)}
\label{L23}
\end{align} and substituting \eqref{L23} into \eqref{L1}, we obtain
\begin{align}
\textbf{W}_{f,1}^{(t)}\bar{\textbf{F}}^{(t)}\textbf{G}_1^{(t)}
+ \textbf{W}_{f,2}^{(t)}\bar{\textbf{F}}^{(t)}\textbf{G}_2^{(t)}
+ (N_rp_r)^{-1}w_f^{(t)}\bar{\textbf{F}}^{(t)}\textbf{G}_r^{(t)}
= \textbf{W}_{f,0}^{(t)}.
\label{L1_2}
\end{align}
To solve \eqref{L1_2}, we use the following lemma from \cite{lemma1_equation}.
\begin{lemma}[\cite{lemma1_equation}]
Let $\textbf{A}_{1,i} \in \mathbb{C}^{N_1 \times N_1}$, $\textbf{A}_{2,i} \in \mathbb{C}^{N_2 \times N_2}$, and $\textbf{A}_3 \in \mathbb{C}^{N_1 \times N_2}$ for $1 \leq i \leq k$.
A matrix $\textbf{X} \in \mathbb{C}^{N_1 \times N_2}$ is a solution of the general linear matrix equation.
\begin{align}
\textbf{A}_{1,1}\textbf{X}\textbf{A}_{2,1} + \textbf{A}_{1,2}\textbf{X}\textbf{A}_{2,2} + \cdots + 
\textbf{A}_{1,k}\textbf{X}\textbf{A}_{2,k} = \textbf{A}_3
\end{align}
if and only if $\textbf{x}=$vec$(\textbf{X})$ is a solution of the equation $\textbf{K}\textbf{x}=\textbf{a}_3$ with $\textbf{K}= \sum_{i=1}^k (\textbf{A}_{2,i}^T \otimes \textbf{A}_{1,i})$  and $\textbf{a}_3=$vec$(\textbf{A}_3)$.
\end{lemma}

Using \emph{Lemma 2},
the optimal $\bar{\textbf{f}}^{(t)^\dag} \left(=\text{vec}(\bar{\textbf{F}}^{(t)^\dag})\right)$ is expressed as
\begin{align}
\bar{\textbf{f}}^{(t)^\dag} = \textbf{W}_f^{(t)^{-1}}\textbf{w}_{f,0}^{(t)}
\end{align}
where
$\textbf{W}_f^{(t)}=( \textbf{G}_1^{(t)^T}\otimes\textbf{W}_{f,1}^{(t)} )
                   +( \textbf{G}_2^{(t)^T}\otimes\textbf{W}_{f,2}^{(t)} )
                   +\{ \textbf{G}_r^{(t)^T}\otimes(N_rp_r)^{-1}w_f^{(t)}\textbf{I}_{N_r} \}$
and
$\textbf{w}_{f,0}^{(t)}=\text{vec}(\textbf{W}_{f,0}^{(t)})$. Note that the optimal $\bar{\textbf{F}}^{(t)}$ in time slot $t$ involves the past channels from time slot $t-m$ to time slot $t$ since it is a function of $ \textbf{G}_1^{(t)}$, $\textbf{G}_2^{(t)}$, and
$ \textbf{G}_r^{(t)} $ in \eqref{G}. Once the optimal $\bar{\textbf{F}}^{(t)}$ is obtained, the optimal $\alpha^{(t)^\dag}$ can be obtained from \eqref{L3} such that 
\begin{align}
\alpha^{(t)^\dag} = \{ N_rp_r 
\text{tr}(\bar{\textbf{F}}^{(t)}\textbf{G}_r^{(t)}\bar{\textbf{F}}^{(t)^H})^{-1} \}^{1/2}
\label{optimal_alpha}
\end{align}
Substituting \eqref{optimal_alpha} into \eqref{L23}, the optimal $\lambda^{(t)^\dag}$ is given by
\begin{align}
\lambda^{(t)^\dag} = (N_rp_r)^{-2}w_f^{(t)} \text{tr}(\bar{\textbf{F}}^{(t)}\textbf{G}_r^{(t)}\bar{\textbf{F}}^{(t)^H})
\label{optimal_lambda}
\end{align} Then, the optimal relay beamforming matrix in time slot $t$ is obtained as $\textbf{F}^{(t)^\dag}=\alpha^{(t)^\dag}\text{mat}(\bar{\textbf{f}}^{(t)^\dag})$. In the next time slot, the optimal relay beamforming matrix is derived in the same way.

\emph{Remark}: The  proposed relay beamforming design method can be applied to one-way FD AF relay systems.
Unlike the two-way FD AF relay systems, using matrix inverse operation, the relay beamforming matrix solution can be obtained without \emph{Lemma 2}.

\subsection{Receive beamforming design at sources}

The problem of receive beamforming matrices design in each time slot is formulated as
\begin{align}
\mathbb{P}_r^{(t)}:
\mathop{\mbox{min}}_{\textbf{R}_1^{(t)},\textbf{R}_2^{(t)}} J(\textbf{R}_1^{(t)},\textbf{R}_2^{(t)})
\label{}
\end{align}
where the sum of MSE $J(\textbf{R}_1^{(t)},\textbf{R}_2^{(t)})$ is defined as
\begin{equation}
\begin{split}
J(\textbf{R}_1^{(t)},\textbf{R}_2^{(t)}) &=
N_s(p_1+p_2)
\\
&- \alpha^{(t)^{-1}} \{
p_2\text{tr}(\textbf{W}_{r,1}^{(t)^H}\textbf{R}_1^{(t)}) +
p_2\text{tr}(\textbf{W}_{r,1}^{(t)}\textbf{R}_1^{(t)^H}) +
p_1\text{tr}(\textbf{W}_{r,2}^{(t)^H}\textbf{R}_2^{(t)}) +
p_1\text{tr}(\textbf{W}_{r,2}^{(t)}\textbf{R}_2^{(t)^H}) \}
\\
&+ \alpha^{(t)^{-2}} \{
\text{tr}(\textbf{W}_{r,3}^{(t)}\textbf{R}_1^{(t)}\textbf{R}_1^{(t)^H}) +
\text{tr}(\textbf{W}_{r,4}^{(t)}\textbf{R}_2^{(t)}\textbf{R}_2^{(t)^H}) \}
\label{Lag2}
\end{split}
\end{equation}
where
\begin{equation}
\begin{split}
& \textbf{W}_{r,1}^{(t)}=\textbf{H}_{r,1}^{(t)}\textbf{F}^{(t)}\textbf{H}_{2,r}^{(t-1)^H}, \\
& \textbf{W}_{r,2}^{(t)}=\textbf{H}_{r,2}^{(t)}\textbf{F}^{(t)}\textbf{H}_{1,r}^{(t-1)^H}, \\
& \textbf{W}_{r,3}^{(t)}=\textbf{H}_{r,1}^{(t)}\textbf{F}^{(t)}\textbf{G}_1^{(t)}\textbf{F}^{(t)^H}\textbf{H}_{r,1}^{(t)^H}+
(N_sp_1\sigma_{e,1}^2+\sigma_{n,1}^2)\textbf{I}_{N_s}, \\
& \textbf{W}_{r,4}^{(t)}=\textbf{H}_{r,2}^{(t)}\textbf{F}^{(t)}\textbf{G}_2^{(t)}\textbf{F}^{(t)^H}\textbf{H}_{r,2}^{(t)^H}+
(N_sp_2\sigma_{e,2}^2+\sigma_{n,2}^2)\textbf{I}_{N_s}.
\end{split}
\end{equation}
By solving $\frac{\partial J(\textbf{R}_1^{(t)},\textbf{R}_2^{(t)})} {\partial \textbf{R}_i^{(t)}}=\textbf{0}_{N_s}$ for $i \in \{ 1,2 \}$, the receive beamforming matrices at the sources are given by 
\begin{equation}
\begin{split}
\textbf{R}_l^{(t)}
= \alpha^{(t)}p_{\bar{l}}
\{\textbf{H}_{r,l}^{(t)}\textbf{F}^{(t)}\textbf{G}_l^{(t)}\textbf{F}^{(t)^H}\textbf{H}_{r,l}^{(t)^H}
+(N_sp_l\sigma_{e,l}^2+\sigma_{n,l}^2)\textbf{I}_{N_s}\}^{-1}
\textbf{H}_{r,l}^{(t)}\textbf{F}^{(t)}\textbf{H}_{\bar{l},r}^{(t-1)}
\end{split}
\end{equation}
where $\bar{l}$ means the index of the other source ($\bar{l}=2$ if $l=1$ and $\bar{l}=1$ if $l=2$).

\subsection{An iterative algorithm for joint beamforming design}

An iterative algorithm is proposed for the joint relay and receive beamforming design.
The relay beamforming matrix and receive beamforming matrices are optimized alternately until convergence in each time slot.
The proposed iterative algorithm is shown in Algorithm 1.
\begin{table}
\centering
\begin{tabular}{l}
\hline
Algorithm 1. The proposed iterative algorithm in time slot $t$ \\ 
\hline
 1:  \hspace{0.15in} Initialize: $\textbf{F}^{(t)}=\textbf{I}_{N_r}$ and $\textbf{R}_i^{(t)}=\textbf{I}_{N_s}$ for $i \in \{ 1,2 \}$ \\
 2:  \hspace{0.15in} Repeat \\
 3:  \hspace{0.35in} Update $\bar{\textbf{F}}^{(t)}$ by using $\textbf{R}_i^{(t)}$ for $i \in \{ 1,2 \}$ \\
 4:  \hspace{0.35in} Update $\alpha^{(t)}$ by using $\bar{\textbf{F}}^{(t)}$ \\
 5:  \hspace{0.35in} Update $\textbf{F}^{(t)}$ by using $\bar{\textbf{F}}^{(t)}$ and $\alpha^{(t)}$ \\
 6:  \hspace{0.35in} Update $\textbf{R}_i^{(t)}$ for $i \in \{ 1,2 \}$ by using $\textbf{F}^{(t)}$ \\
 7:  \hspace{0.35in} Update $J(\textbf{F}^{(t)},\alpha^{(t)},\textbf{R}_1^{(t)},\textbf{R}_2^{(t)})$ \\
 8:  \hspace{0.15in} Until convergence. \\
\hline
\end{tabular}
\end{table}

The achievable sum rate in time slot $t$ is given by
\begin{equation}
\begin{split}
R^{(t)}
= \mathbb{E}\big[ 
\text{log}_2 \text{det} \{ \textbf{I}_{N_s} +
p_1 \textbf{R}_2^{(t)^H}\textbf{H}_{r,2}^{(t)}\bar{\textbf{F}}^{(t)}\textbf{H}_{1,r}^{(t-1)}
\textbf{H}_{1,r}^{(t-1)^H}\bar{\textbf{F}}^{(t)^H}\textbf{H}_{r,2}^{(t)^H}\textbf{R}_2^{(t)}
(\textbf{R}_2^{(t)^H}\textbf{A}_2^{(t)}\textbf{R}_2^{(t)})^{-1} \}
\big]
\\
+ \mathbb{E}\big[
\text{log}_2 \text{det} \{ \textbf{I}_{N_s} +
p_2 \textbf{R}_1^{(t)^H}\textbf{H}_{r,1}^{(t)}\bar{\textbf{F}}^{(t)}\textbf{H}_{2,r}^{(t-1)}
\textbf{H}_{2,r}^{(t-1)^H}\bar{\textbf{F}}^{(t)^H}\textbf{H}_{r,1}^{(t)^H}\textbf{R}_1^{(t)}
(\textbf{R}_1^{(t)^H}\textbf{A}_1^{(t)}\textbf{R}_1^{(t)})^{-1} \}
\big]
\end{split}
\end{equation}
where $\textbf{A}_l^{(1)} = \sigma_{n,r}^2\textbf{H}_{r,l}^{(1)}\bar{\textbf{F}}^{(1)}\bar{\textbf{F}}^{(1)^H}\textbf{H}_{r,l}^{(1)^H}
+ \alpha^{(1)^{-2}}p_l\boldsymbol{\Delta}_{l,l}^{(1)}\boldsymbol{\Delta}_{l,l}^{(1)^H} + \alpha^{(1)^{-2}}\sigma_{n,l}^2\textbf{I}_{N_s}$
and
\begin{equation}
\begin{split}
\textbf{A}_l^{(t)}
&= \sigma_{n,r}^2\textbf{H}_{r,l}^{(t)}\bar{\textbf{F}}^{(t)}\bar{\textbf{F}}^{(t)^H}\textbf{H}_{r,l}^{(t)^H}
+ \sum_{i=2}^{t}  \bigg\{  \textbf{H}_{r,l}^{(t)}\bar{\textbf{F}}^{(t)} \prod_{j=1}^{i-1}( \boldsymbol{\Delta}_{r,r}^{(t-j)}\textbf{F}^{(t-j)})
\\
&(p_l\textbf{H}_{l,r}^{(t-i)}\textbf{H}_{l,r}^{(t-i)^H} + p_{\bar{l}}\textbf{H}_{\bar{l},r}^{(t-i)}\textbf{H}_{\bar{l},r}^{(t-i)^H} + \sigma_{n,r}^2\textbf{I}_{N_r})
\prod_{j=t+1-i}^{t-1}(\textbf{F}^{(j)^H}\boldsymbol{\Delta}_{r,r}^{(j)^H})
\bar{\textbf{F}}^{(t)^H}\textbf{H}_{r,l}^{(t)^H}
\bigg\}
\\
&+ \alpha^{(t)^{-2}}p_l\boldsymbol{\Delta}_{l,l}^{(t)}\boldsymbol{\Delta}_{l,l}^{(t)^H}
+ \alpha^{(t)^{-2}}\sigma_{n,l}^2\textbf{I}_{N_s}
\end{split}
\end{equation}
and $\bar{l}$ means the index of the other source ($\bar{l}=2$ if $l=1$ and $\bar{l}=1$ if $l=2$).

The sum of MSE $J_{m}^{(t)}=J(\textbf{F}^{(t)},\alpha^{(t)},\textbf{R}_1^{(t)},\textbf{R}_2^{(t)})$ depends on $m$ and $t$.
When $m=\infty$, $J_{m}^{(t)}$ increases or does not change as $t$ increases.
When $m$ is a finite number, $J_{m}^{(t)}$ increases or does not change as $t$ increases only if $t \leq m+1$, but
$J_{m}^{(t)}$ decreases as $t$ increases if $t \geq m+2$ because the relay substitutes the oldest beamforming matrix with a new one.
In other words, the MSE of system can be unstable after time slot $m+2$. Based on this observation, we propose an algorithm to determine $m$ in Algorithm 2, where the determined $m$ is denoted by $\hat{m}$. For example, if $J_{m=2}^{(3)} \leq J_{m=2}^{(4)}$, we set $\hat{m}$ to be 2 for stability. 

\begin{table}
\centering
\begin{tabular}{l}
\hline
Algorithm 2. The proposed algorithm for obtaining $\hat{m}$ \\ 
\hline
1: \hspace{0.15in} Initialize: $i=0$ \\
2: \hspace{0.15in} Repeat \\
3: \hspace{0.35in} $i \leftarrow i+1$ \\
4: \hspace{0.35in} If $J_{m=i}^{(i+1)} \leq J_{m=i}^{(i+2)}$ \\
5: \hspace{0.55in} $\hat{m}=i$ \\
6: \hspace{0.35in} End \\
7: \hspace{0.15in} Until $\hat{m}$ for stability is determined. \\
\hline
\end{tabular}
\end{table}

\section{Numerical results}

In this section, we numerically evaluate the sum of MSE and the achievable sum rate for the proposed scheme, conventional scheme, and the HD two-way relaying scheme.
The proposed scheme updates the relay beamforming matrix $\textbf{F}^{(t)}$ and receive beamforming matrix $\textbf{R}_i^{(t)}$ for $i \in \{ 1,2 \}$ every time slot, which is based on the channels in the $m$ latest time slots. Unlike the proposed scheme, the conventional scheme updates the relay beamforming matrix based only on the channels in the current time slot,
i.e. $\textbf{G}_c^{(t)}=\textbf{0}_{N_r}$. The HD two-way relaying scheme is used as a referential scheme. 
The sum of MSE and achievable sum rate are obtained by averaging over 100 channel realizations.
The number of iteration for relay and receive beamforming design is 30.

The channel matrices $\textbf{H}_{i,j}^{(t)}$ for $(i,j) \in \{ (1,r),(r,1),(2,r),(r,2) \}$ are set to follow Rayleigh fading, i.e., the elements of each channel matrix are independent complex Gaussian random variables with zero mean and unit variance.
The channel estimation error matrices $\boldsymbol{\Delta}_{i,i}^{(t)}$ for $i \in \{ 1,2,r \}$ are set to follow Rayleigh fading, i.e., the elements of each channel matrix are independent complex Gaussian random variables with zero mean and variance $\sigma_{e,i}^2$.
For simplicity, it is assumed that all the nodes transmit with power $p$, i.e., the power is set to be $p_1=p_2=p_r=p=1$.
The noise variance and the loopback channel estimation error variance at each node are set to be as $\sigma_{n,1}^2=\sigma_{n,2}^2=\sigma_{n,r}^2=\sigma_n^2$ and $\sigma_{e,1}^2=\sigma_{e,2}^2=\sigma_{e,r}^2=\sigma_e^2$.
SNR and INR are defined as $p/\sigma_n^2$ and $\sigma_e^2/\sigma_n^2$, respectively.

\begin{figure} [th!]
\centering
\subfigure[Sum of MSE versus number of iterations for different INR values.]
{
\includegraphics[width=10cm]{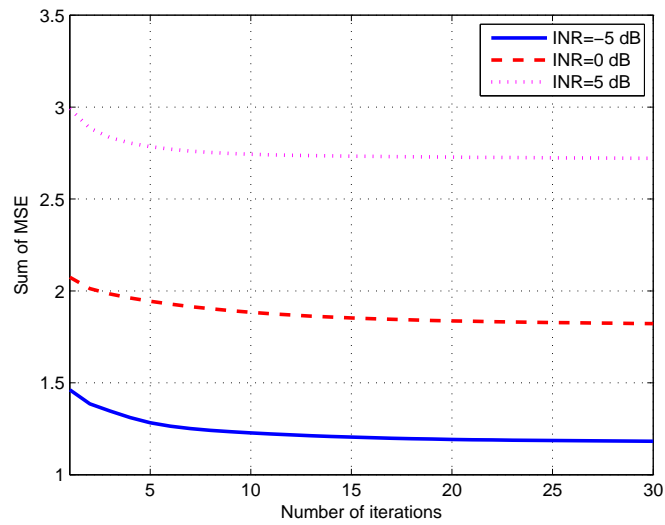}
\label{fig_result_iteration_MSE}
}
\subfigure[Achievable sum rate versus number of iterations for different INR values.]
{
\includegraphics[width=10cm]{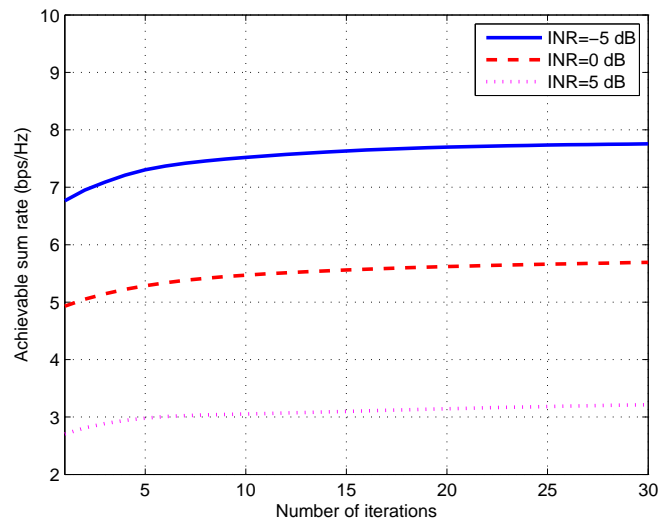}
\label{fig_result_iteration_rate}
}
\caption{Performance of the proposed scheme versus number of iterations for different INR}
\label{fig_result_iteration}
\end{figure}
Fig. \ref{fig_result_iteration} shows the sum of MSE and the achievable sum rate versus number of iterations for different INR values in the 10th time slot when SNR$=5$dB, $N_s=2$, $N_r=5$, and $m=\infty$, respectively.
The proposed scheme addresses two sub-problems for relay beamforming design and receive beamforming design at sources.
Since the solution for each sub-problem is optimal, the sum of MSE  decreases with each iteration.

\begin{figure} [t!]
\centering
\subfigure[Sum of MSE versus SNR for different INR values.]
{
\includegraphics[width=10cm]{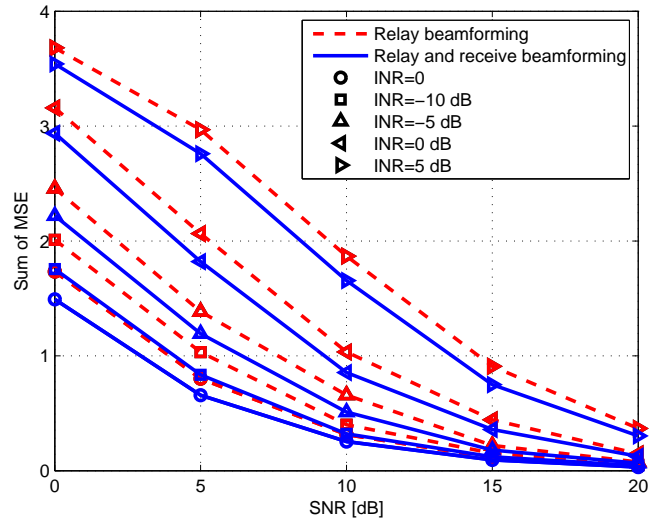}
\label{fig_result_SNR_pro_MSE}
}
\subfigure[Achievable sum rate versus SNR for different INR values.]
{
\includegraphics[width=10cm]{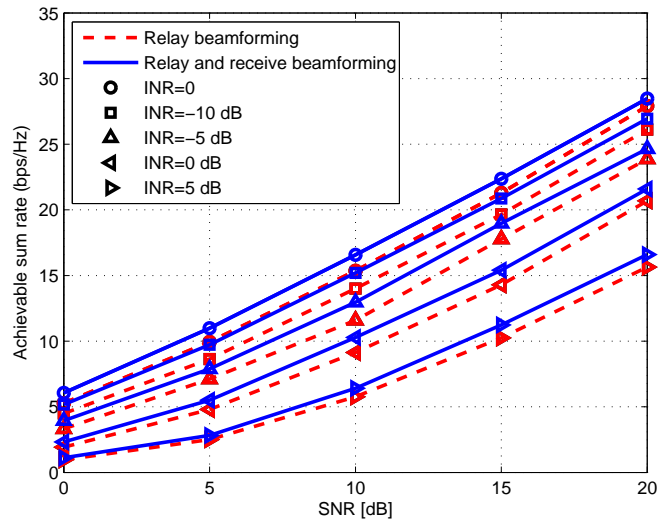}
\label{fig_result_SNR_pro_rate}
}
\caption{Performance comparisons between the proposed relay beamforming and the proposed joint beamforming schemes versus SNR}
\label{fig_result_SNR_pro}
\end{figure}
Fig. \ref{fig_result_SNR_pro} shows the sum of MSE and achievable sum rate versus SNR for different INR values in the 10th time slot when $N_s=2$, $N_r=5$, and $m=\infty$, respectively.
As shown, the proposed joint beamforming design achieves better performance than  the proposed design of relay beamforming only.
As SNR increases and INR decreases, the sum of MSE decreases and the achievable sum rate increases for both the proposed schemes.

\begin{figure} [!t]
\centering
\subfigure[Sum of MSE versus time slot index for different values of $m$ and INR]
{
\includegraphics[width=10cm]{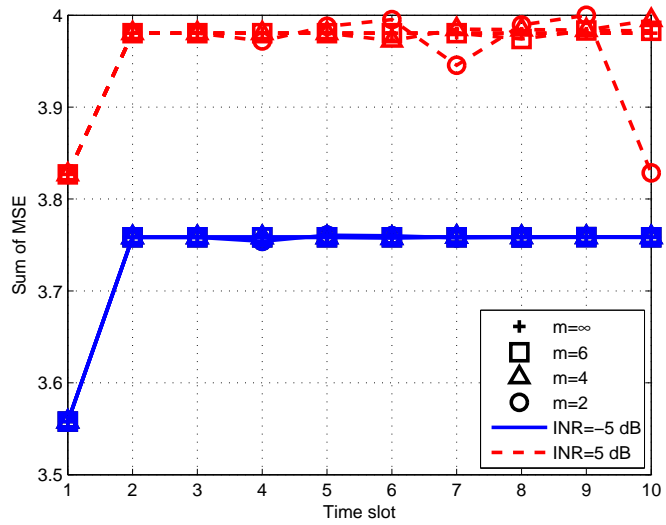}
\label{fig_result_m_MSE}
}
\subfigure[Achievable sum rate versus time slot index for different values of $m$ and INR]
{
\includegraphics[width=10cm]{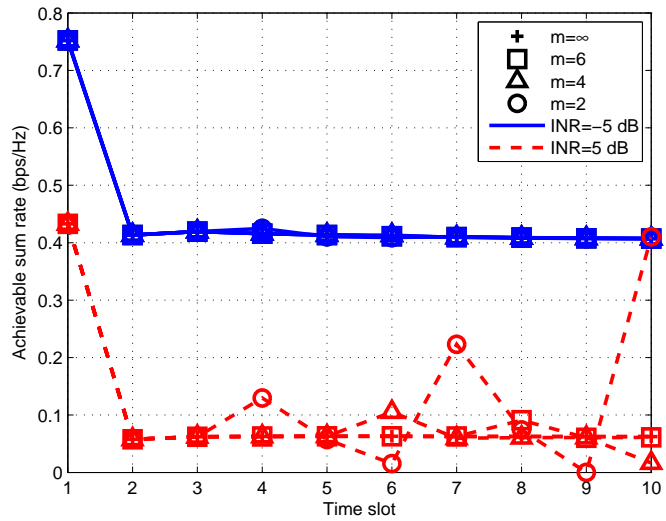}
\label{fig_result_m_rate}
}
\caption{Performance of the proposed joint beamforming design scheme versus time slot index for different values of $m$ and INR}
\label{fig_result_m}
\end{figure}

Fig. \ref{fig_result_m} shows the sum of MSE and the achievable sum rate versus the time slot index for different values of $m$ and INR when SNR$=-10$ dB, $N_s=2$, and $N_r=5$.
The proposed scheme with $m=\infty$ and $m=\hat{m}$ are the methods that the relay and receive beamforming matrices are designed with all the beamforming and channel matrices from the first time slot and with only those in the $\hat{m}$ latest time slots, respectively.
If $\hat{m}$ is too small, the performance is unstable after time slot $\hat{m}+2$ because it cannot handle the propagated effect of the residual loopback SI.
For instance, when $m=2, 4$ under INR$=-5$ dB and $m=2, 4, 6$ under INR$=5$ dB, both of the sum MSE and the sum rate oscillate and the oscillations become larger as either INR increases or $m$ decreases.
On the other hand, if $m \geq 6$ under INR$=-5$ dB, the performance gap from the case of $m=\infty$ is marginal.

\begin{figure} [t!]
\centering
\includegraphics[width=10.0cm]{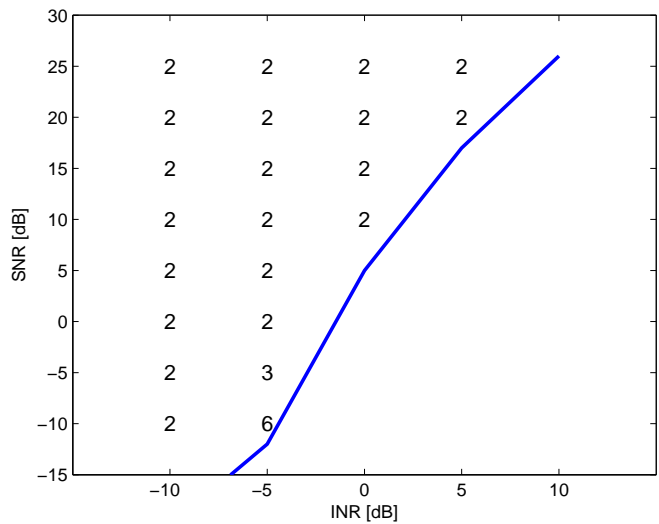}
\caption{The criterion for duplex mode selection and $\hat{m}$ for stability in the FD mode.}
\label{fig_result_range}
\end{figure}

Fig. \ref{fig_result_range} shows the criteria for the duplex mode selection and the selection of $\hat{m}$ for stability in the FD mode when $N_s=2$ and $N_r=5$.
In the region above the solid line, the achievable sum rate of the proposed scheme in FD mode is greater than that of HD mode, so FD is preferred to HD. However, below the solid line, HD is better than FD. This figure also reveals that as INR increases, the SNR required for selection of FD grows. For each pair of SNR and INR, the values of $\hat{m}$ for stability are presented, too. For the region where HD is preferred, $\hat{m}$ is not specified because the residual loopback SI in the HD mode does not exist.

\begin{figure} [!t]
\centering
\subfigure[Sum of MSE versus SNR for different INR values.]
{
\includegraphics[width=10cm]{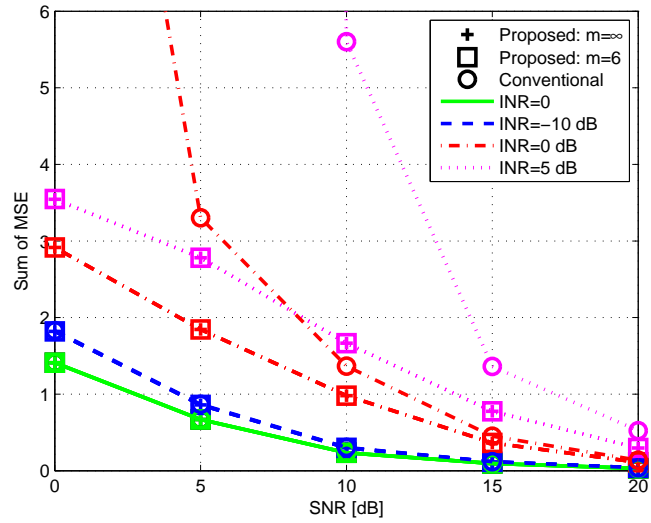}
\label{fig_result_SNR_procon_MSE}
}
\subfigure[Achievable sum rate versus SNR for different INR values.]
{
\includegraphics[width=10cm]{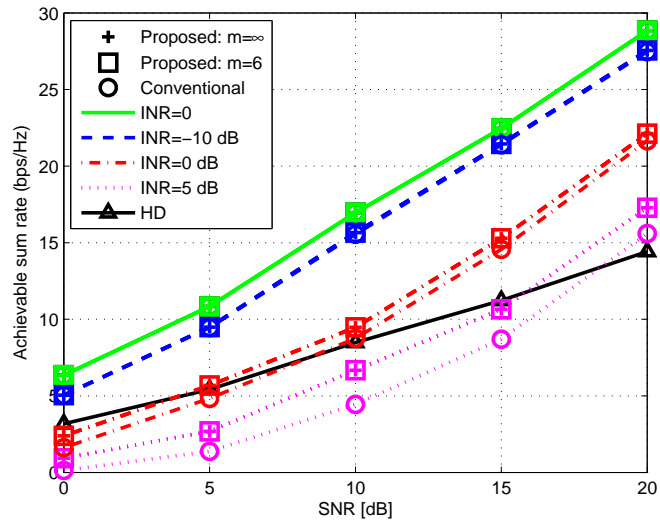}
\label{fig_result_SNR_procon_rate}
}
\caption{Performance comparison between the proposed scheme of joint beamforming design and the conventional scheme versus SNR}
\label{fig_result_SNR_procon}
\end{figure}

Fig. \ref{fig_result_SNR_procon_MSE} shows the sum of MSE versus SNR for different INR values in the 10th time slot when $N_s=2$ and $N_r=5$. It is shown that the proposed scheme of joint beamforming design provides a lower MSE than the conventional scheme.
As SNR increases and INR decreases, the MSE decreases for both the proposed and conventional schemes. The performance difference between the two schemes decreases as SNR increases and INR decreases.
This is because the proposed scheme updates the relay and receive beamforming matrices based on the covariance matrix of loopback channel estimation error.
This figure also verifies that $m=6$ provides almost the same MSE performance as $m=\infty$.
Fig. \ref{fig_result_SNR_procon_rate} compares the achievable sum rates of the proposed scheme of joint beamforming design and the conventional scheme at the 10th time slot when $N_s=2$ and $N_r=5$.
In either low SNR or high INR regime, FD is not beneficial relative to HD because the cost of handling the residual loopback SI is higher than the HD loss in those regimes, which strongly suggests adaptive selection between HD and FD according to SNR and INR; if INR is less than $0$ dB when SNR $\geq$ $5$ dB or less than $5$ dB when SNR $\geq$ $16$ dB, FD with the proposed beamforming scheme is preferred to HD.
If we focus on the scenarios where FD is preferred, $m=6$ offers almost the same achievable sum rate as $m=\infty$, as in Fig. \ref{fig_result_SNR_procon_MSE}.

\begin{figure}[t]
\centering
\includegraphics[width=11cm]{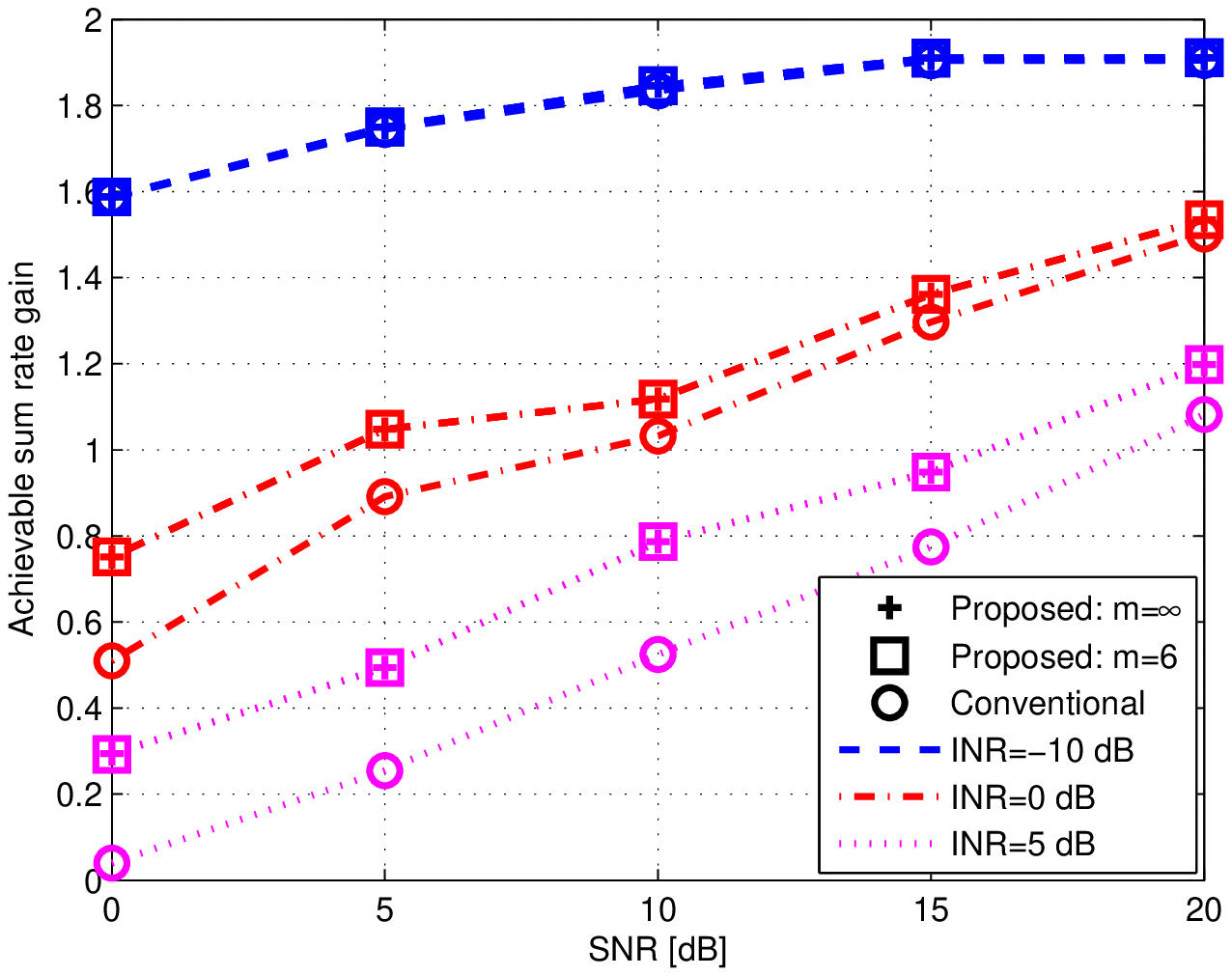}
\caption{Achievable sum rate gain versus SNR for different INR values.}
\label{fig_result_SNR_procon_gain}
\end{figure}

Fig. \ref{fig_result_SNR_procon_gain} shows the achievable sum rate gain versus SNR for different INR values at the 10th time slot when $N_s=2$ and $N_r=5$. The achievable sum rate gain is defined as the ratio between the achievable sum rate of the proposed joint design scheme and that of the HD scheme. For  example, the achievable sum rate gain is 2 and 1.9 when INR$=0$ and INR is $-10$ dB, respectively. The gain increases as SNR decreases or INR increases.

\begin{figure} [!t]
\centering
\subfigure[Sum of MSE versus time slot index]
{
\includegraphics[width=10cm]{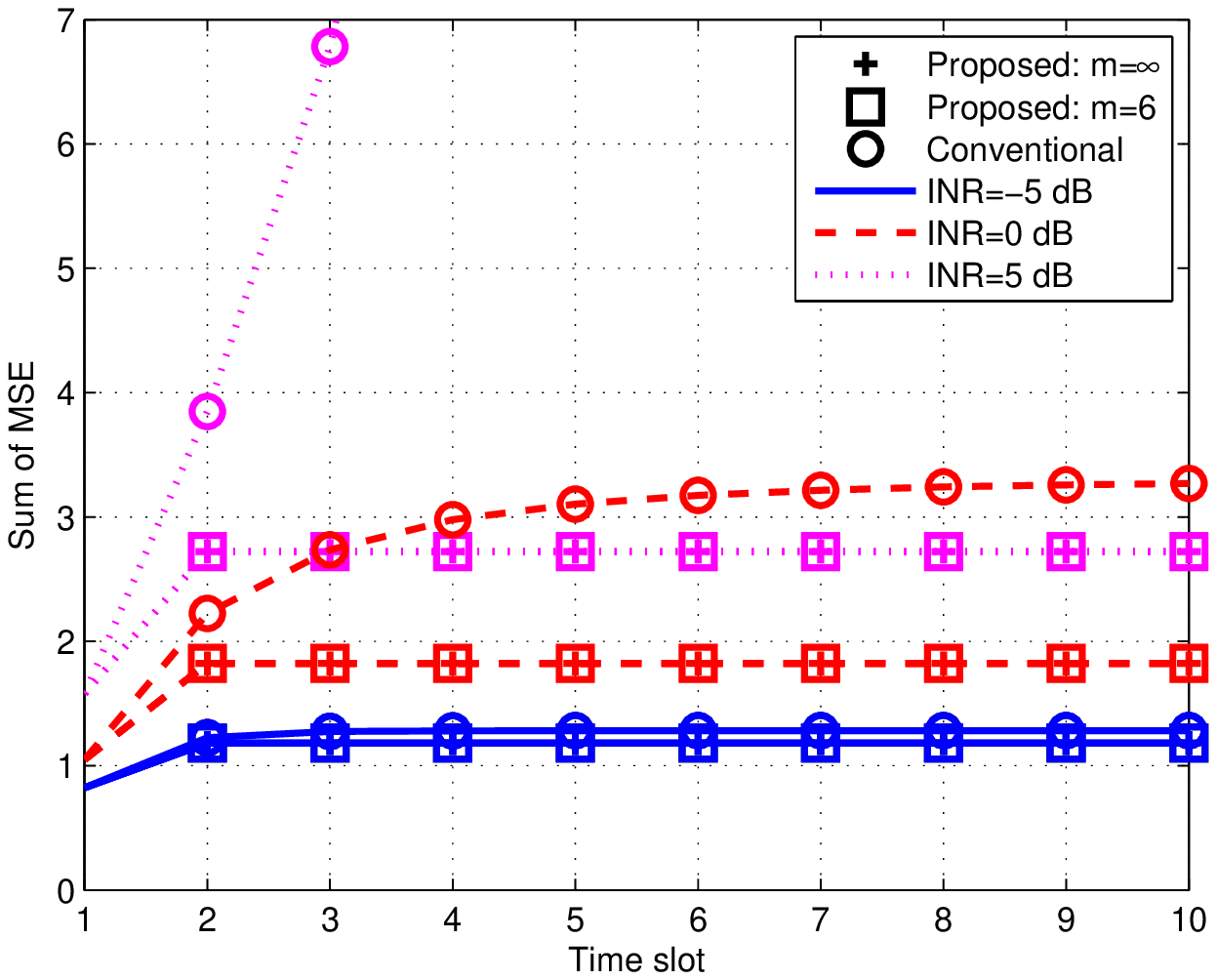}
\label{fig_result_ts_MSE}
}
\subfigure[Achievable sum rate versus time slot index]
{
\includegraphics[width=10cm]{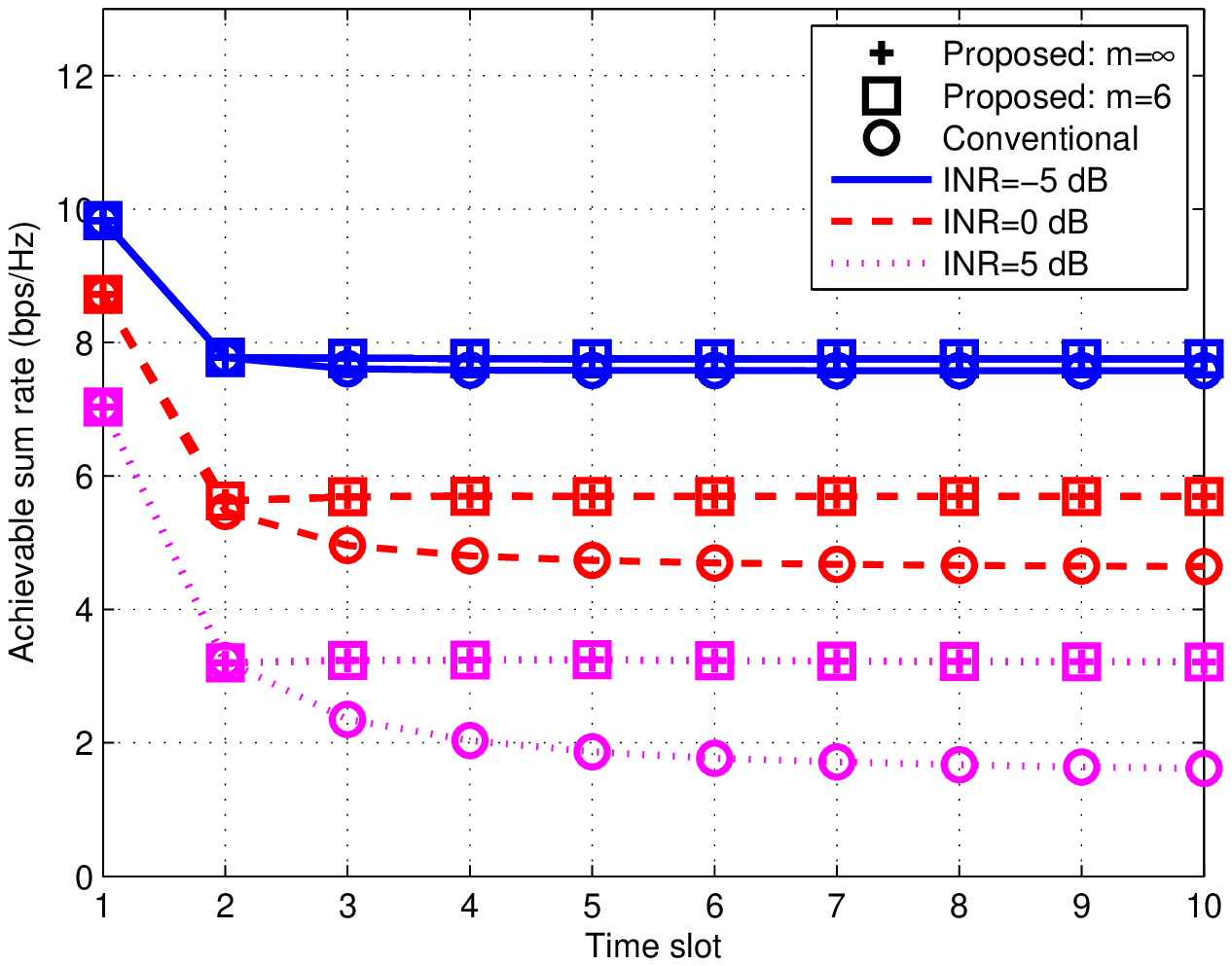}
\label{fig_result_ts_rate}
}
\caption{Performance comparison between the proposed and conventional schemes versus time slot index}
\label{fig_result_ts}
\end{figure}

\begin{figure} [!t]
\centering
\subfigure[Sum of MSE versus SNR for different values of $N_r$]
{
\includegraphics[width=10cm]{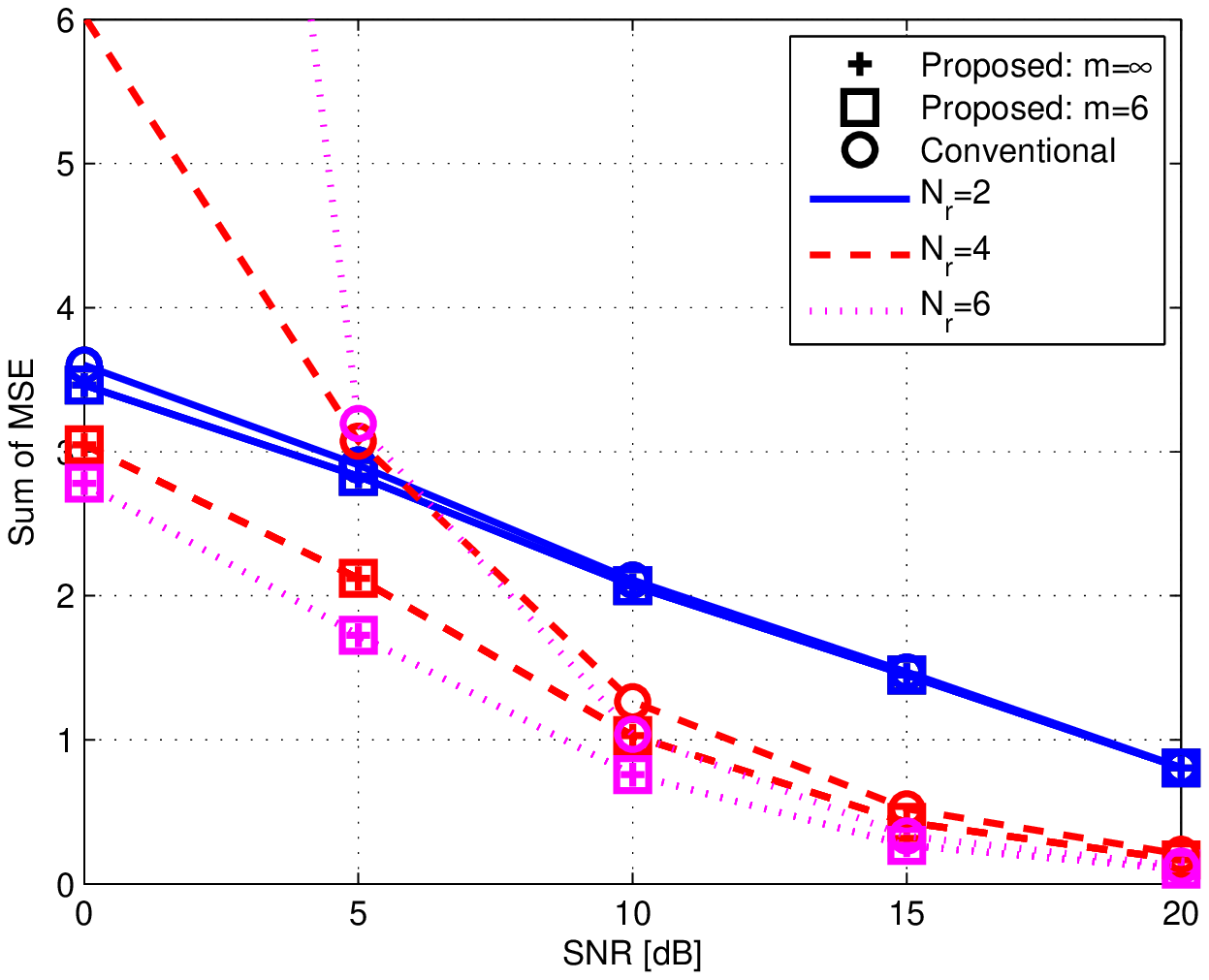}
\label{fig_result_Nr_MSE}
}
\subfigure[Achievable sum rate versus SNR for different values of $N_r$]
{
\includegraphics[width=10cm]{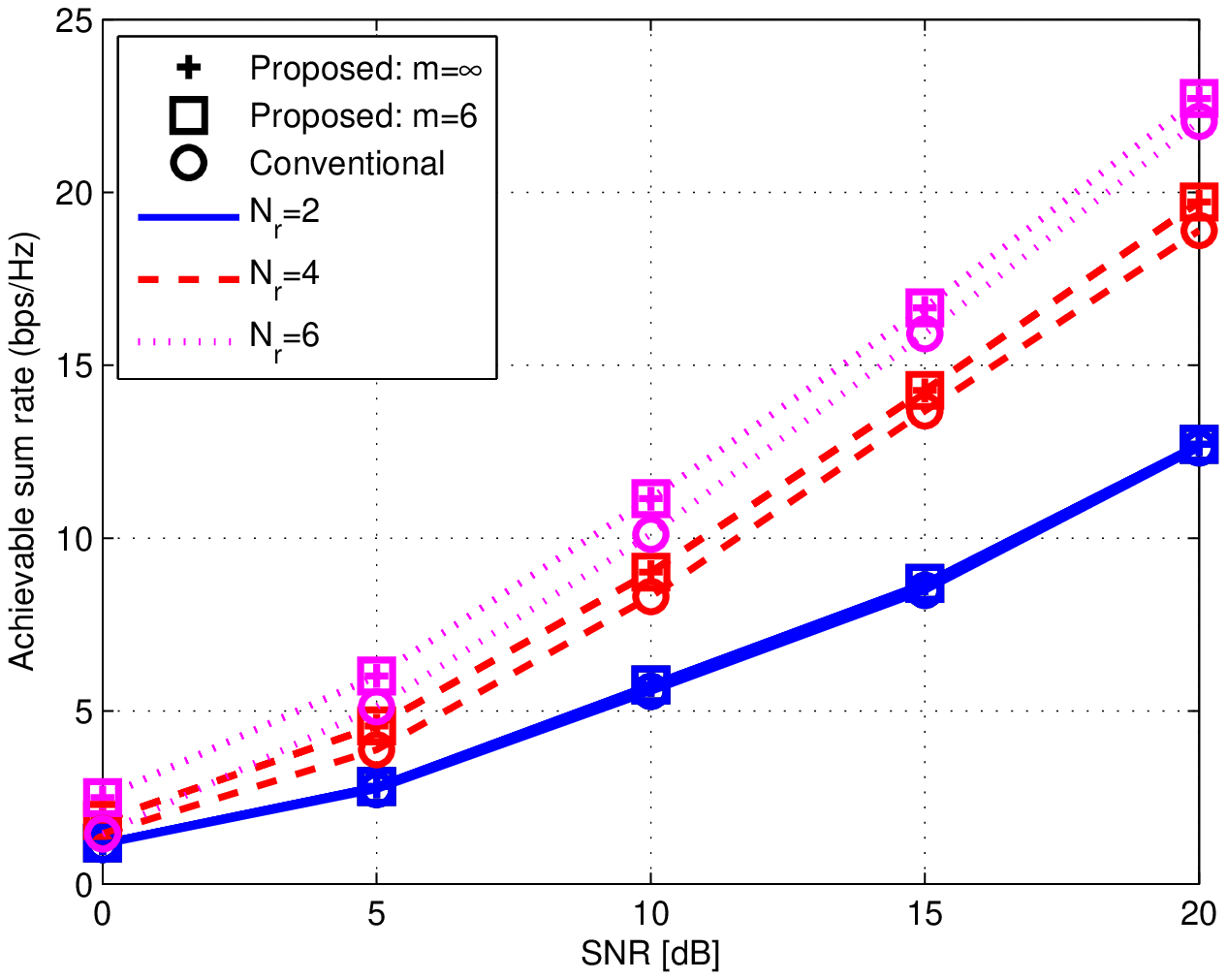}
\label{fig_result_Nr_rate}
}
\caption{Performances comparison between the proposed and conventional schemes versus SNR}
\label{fig_result_Nr}
\end{figure}

Fig. \ref{fig_result_ts} shows the sum of MSE and achievable sum rate versus the time slot index for different INR (i.e., INR$=-5$ dB, $0$ dB, and $5$ dB) when SNR$=5$ dB, $N_s=2$, and $N_r=5$. The performance gap between the proposed joint design scheme and the conventional scheme is shown to be larger as the time slot index increase, for both MSE and sum rate. The error propagation effect of residual loopback SI over time slots is not properly addressed in the conventional scheme and accordingly the loss due to the error propagation increases with time slot index. Interestingly, the performance gap between the first and the second time slots is greater than that between the second and the third time slots. This is because the residual SI at relay starts to affect from the second time slot whereas the residual SI at source affects from the first time slot. As in the previous figures, it is verified again that $m=6$ offers almost the same performance as $m=\infty$.

Fig. \ref{fig_result_Nr} presents the sum of MSE and achievable sum rate versus SNR for different values of $N_r$ at the $10$th time slot when $N_s=2$ and INR$=0$ dB.
As the $N_r$ increases, the MSE decreases and achievable sum rate increases for both the proposed joint design scheme and the conventional scheme. The performance gap between the proposed and conventional schemes increases with $N_r$ because the multiple antennas at relay are more efficiently used in the proposed scheme.

\section{Conclusion}

This paper investigated joint design of relay beamforming and receive beamforming at sources in  FD two-way AF relay systems under a relay transmit power constraint. We analyzed the coupled effect of beamforming matrix design across time slots and derived the optimal bemforming matrices in closed form at each time slot in terms of MMSE, considering the coupled effect. To reduce the burden on the beamforming design incorporating all the time slots from the first time slot to the current time slot,  we also proposed the beamforming design method based only on the $m$ latest time slots. It was shown that if $m$ is not too small, the performance degradation due to a limited number of $m$ was marginal. With numerical analysis, we revealed when full-duplex operation is beneficial compared to HD operation, and proposed adaptive selection between FD and HD according to INR and SNR.

\end{document}